\documentclass[smallcondensed]{svjour3}
\usepackage{amssymb,amsmath}
\usepackage{graphicx,color}
\usepackage{graphics}

\def\ve {\varepsilon}
\def\eps {\varepsilon}

\usepackage[colorlinks]{hyperref}
\hypersetup{
    linkcolor=blue,
}

\title{Motor Protein Transport Along Inhomogeneous Microtubules}

\author{S.~D.~Ryan \and Z. McCarthy \and M. Potomkin*}

\institute{S.~D. Ryan
			\at Department of Mathematics and Statistics, Cleveland State University, Cleveland, OH 44115, USA\\
			Center for Applied Data Analysis and Modeling, Cleveland State University, Cleveland, OH 44115, USA\\ \email{s.d.ryan@csuohio.edu}
			\and
			Z. McCarthy
			\at Department of Mathematics and Statistics, York University, Toronto, ON, Canada\\
			 Laboratory for Industrial and Applied Mathematics, Toronto, ON, Canada\\
			 Centre for Disease Modelling, York University, Toronto, ON, Canada\\
			 Fields-CQAM Mathematics for Public Health Laboratory, Toronto, ON, Canada
			\and 
			M. Potomkin
			\at Department of Mathematics, University of California,  Riverside, CA, 92521, USA\\
			*Corresponding Author \email{mykhailp@ucr.edu}
			}

\begin{document}

\authorrunning{S.~D. Ryan et al.}

\maketitle

\vspace{-.1 in}

 \begin{abstract}
Many cellular processes rely on the cell's ability to transport material to and from the nucleus. Networks consisting of many microtubules and actin filaments are key to this transport. Recently, the inhibition of intracellular transport has been implicated in neurodegenerative diseases such as Alzheimer's disease and Amyotrophic Lateral Sclerosis (ALS). 
Furthermore, microtubules may contain so-called {\it defective regions} where motor protein velocity is reduced due to accumulation of other motors and microtubule associated proteins. In this work, we propose a new mathematical model describing the motion of motor proteins on microtubules which incorporate a defective region.  We take a mean-field approach derived from a first principle lattice model to study motor protein dynamics and density profiles.
 In particular, given a set of model parameters we obtain a closed-form expression for the equilibrium density profile along a given microtubule.   We then verify the analytic results using mathematical analysis on the discrete model and Monte Carlo simulations.  This work will contribute to the fundamental understanding of inhomogeneous microtubules providing insight into microscopic interactions that may result in the onset of neurodegenerative diseases. Our results for inhomogeneous microtubules are consistent with prior work studying the homogeneous case.
 \end{abstract}
 



 \section{Introduction}

Cells strongly rely on the ability to efficiently transport cargo via motor proteins.  Fast active transport (FAT) is required to deliver materials such as proteins, mRNA, mitochrondria, vescicles, and organelles for use in a variety of cellular processes \cite{Ros08,Lakadamyali2014}. One of the three major means of intracellular transport is within networks of microtubules (MTs) \cite{Lakadamyali2014}.  To accommodate FAT, motor proteins ``walk" along MTs and actin filaments (AFs) with their cargo forming a ``superhighway" for cellular transportation \cite{Klu05b,Ros14,Str20}.  Motor-protein families kinesin and dynein bind to the cellular material or cargo they carry as they move up and down the microtubule \cite{Gro07,Lakadamyali2014}.  For a more general overview of molecular motor protein motion see review \cite{Mal04}.

Defects in microtubules are known to exist, but the current literature has yet to clarify their impact on molecular motor-based transport \cite{Lia16}.  Defects in active transport, particularly axonal, have been implicated in the progression of various diseases \cite{Lakadamyali2014}. For instance, the defining characteristics of many neurodegenerative diseases such as Alzheimer's disease and Amyotrophic Lateral Sclerosis (ALS) may be related to deficiencies in active transport within neurons \cite{Lakadamyali2014}. One such area of need for immediate study is the scenario where the microtubule paths used by motor proteins become congested, obstructed, or defective.  Hallmarks and early indicators of neurodegenerative diseases are an accumulation of organelles and proteins in the cell body or axon, which inhibits active transport \cite{Lakadamyali2014}. Hence, understanding the nature of motor protein dynamics will provide insight in understanding the onset of these diseases and developing control strategies. 


Advances in biophysical tools and imaging technology have allowed for many recent insightful \textit{in vitro} experiments of motor protein behavior on microtubules \cite{Ros08,Lakadamyali2014,Lia16}. Motor-proteins may change directions, stop or pause briefly, increase their velocity, and also attach/detach from the microtubule \cite{Ros08}. Experimentally it is observed that this behavior may be attributed to the presence of MT-associated proteins \cite{Ros08}. Therefore, the cell's ability to regulate active transport may be studied through tau MAP regions where motion of motor proteins is inhibited. We refer to these patches of high tau MAP concentration as defective regions due to their effects on the reduction of motor motility. We focus our study to modeling collective motor protein motility on MTs, first homogeneous and then defective. 

Modeling motor protein motility on MTs has received recent theoretical attention. A one-dimensional discrete lattice model was studied in detail using a mean-field approach, including a full phase diagram for the stationary states \cite{ParFraFre2003,ParFraFre2004,Fre11,Klu08,Klu05}. The model was capable of predicting the emergence of interior layers by splitting the equilibrium density of motor proteins along the MT into two phases: low and high density. Such a co-existence corresponds to a traffic jam consisting of motor proteins translating along the MT in one direction, from the region with low density to the one with high density. In addition, a generalization of this lattice gas model has also been studied to account for local interactions between motors \cite{Fre18}. In particular, the effects of adjacent motors enhancing the detachment rate as well as motor crowding enhancing the frequency which motors become inactive or paused \cite{Fre18}. In a similar light, models featuring multiple ``lanes" on a MT and motor proteins may switch lanes have also been studied \cite{Fre07} or between a filament and a tube in \cite{Mul05}. The common feature among these lattice gas models is they follow the totally asymmetric simple exclusion process (TASEP) framework \cite{ParFraFre2003}. These models also feature attachment and detachment processes for motors whose stationary states are described by the theory of Langmuir dynamics \cite{ParFraFre2003}. 



 In this paper, we seek to advance the current understanding of subcellular transport along microtubules with a defective region. 
 In Section~\ref{sec:model}, we introduce a new discrete model for motor protein motility on a microtubule with a one-dimensional lattice following the TASEP paradigm and attachment/detachment dynamics. The distinctive feature of this model is that the MT consists of a defective region with decreased motor protein motility rate. From a discrete lattice formulation, we take the limit to recover the mean-field form of these equations. Next, in Section~\ref{sec:homogeneous}, we verify that our mean-field approximations produce results consistent with existing studies (e.g., \cite{ParFraFre2003,ParFraFre2004,Fre11}) before moving to the defective region case where motility will be hindered. The main analytical results are then presented in Section~\ref{sec:nonhomo}, where we consider examples with the domain split into two regions, fast and slow, whereas the local and boundary attachment/detachment rates are varied as parameters.  From these studies we find a closed-form expression of the solution given a set of the model parameters for attachment, detachment, and boundary conditions.  We then verify that the analytical solution of the mean-field model is consistent with corresponding Monte Carlo simulations of the original discrete lattice model in Section~\ref{sec:num}. We note that if one needs to consider a wide range of problem parameters, for example, in model calibration or a control problem (e.g., finding the location and width of a defective region for a desired motor distribution), Monte Carlo simulations 
 	are prohibitively time consuming as compared to a closed-form expression when available. In the Appendix, we provide a new analytical approach to the solution of the homogeneous problem, based on the analysis of the phase portrait of the corresponding system of ODEs; then we give examples of applications of the approach for specific problem parameters. 
 	Overall, this work provides a critical result for inhomogeneous MTs consisting of multiple parts with different motility properties; namely, it can be modeled as segments of homeogeneous MTs linked by a matching flux condition.
 	This greatly expands the utility of past studies that developed the theory of homogeneous MTs.

\section{Mathematical Model} \label{sec:model}

\subsection{Discrete problem}
Following the general TASEP paradigm, we construct a discrete model from first principles which generalizes previous models for homogeneous MTs to incorporate a defective region. Briefly, the TASEP paradigm states: \textit{(i)} each binding site may be occupied by a maximum of one motor; \textit{(ii)} motors move unidirectionally on the lattice and \textit{(iii)} motors enter the lattice on the left side and exit the lattice on the right side \cite{ParFraFre2003}. Here we also account for the attachment and detachment of motor proteins on the MT interior as in works similar in scope focused on modeling \cite{ParFraFre2003,ParFraFre2004,Fre11} and experiment \cite{Var09}.  Another recent work has focused on stochastic modeling with the goal of revealing how motor protein and filament properties affect transport~\cite{Dal19}.

Specifically, consider a one-dimensional lattice $\left\{0=x_0<x_1<...<x_M=\ell\right\}$, representing sites which a motor may occupy on a microtubule of length $\ell$. Introduce $\rho_{i}^{n}$ which is the probability of finding a motor at site $x_i$ at time step $n$. Probabilities at each lattice site $0\leq \rho_i^{n}\leq 1$ for $1\leq i <M$ change during one time step via    
\begin{equation}\label{original_system} 
\rho_i^{n+1}-\rho_i^{n} = v_{i-\frac{1}{2}} \rho_{i-1}^{n} (1-\rho_i^{n})-v_{i+\frac{1}{2}}\rho_i^{n}(1-\rho_{i+1}^{n})+\omega_{A}  (1-\rho_i^{n})-\omega_D  \rho_i^{n}.
\end{equation}
The first term on the right-hand side of \eqref{original_system} says that the probability of finding a motor at site $x_i$ increases due to a possible jump of a motor from site $x_{i-1}$ to $x_i$ provided that the following jump condition is satisfied: there is a motor at site $x_{i-1}$ and site $x_i$ is vacant. As it is done in previous works on one-dimensional transport of active motors along a microtubule \cite{ParFraFre2003,ParFraFre2004,Fre11}, consider the problem in the mean-field approximation, that is, correlations are negligibly small and the probability of the jump condition is simply  $\rho_{i-1}^{n} (1-\rho_i^{n})$. Additionally, the coefficient $v_{i-\frac{1}{2}}$ accounts for inhomogeneity of the microtubule: if the jump condition is satisfied on the interval $[x_{i-1},x_i]$, then the jump occurs with probability (motility rate) $v_{i-\frac{1}{2}}$, and these coefficient may change from site to site. 
The second term on the right-hand side of \eqref{original_system} is similar to the first term, but accounts for the decrease in probability $\rho_i^{n}$ due to a possible jump from site $x_{i}$ to $x_{i+1}$.  
The third term on the right-hand side of \eqref{original_system} describes the interaction of the microtubule with the exterior environment: a motor from outside can attach to the microtubule at site $x_i$ as well as a motor already occupying site $x_i$ can detach from the microtubule. Parameters $\omega_{A}$ and $\omega_{D}$ are the corresponding attachment and detachment rates.  

Stationary states of \eqref{original_system} solve the following system:
 \begin{equation}\label{original_system_stationary} 
 0 = v_{i-\frac{1}{2}} \rho_{i-1} (1-\rho_i)-v_{i+\frac{1}{2}}\rho_i(1-\rho_{i+1})+\omega_{A}  (1-\rho_i)-\omega_D  \rho_i.
\end{equation}   
This system is supplemented with boundary conditions corresponding to the attachment rate $\alpha$ at the left end and detachment rate $\beta$ at the right end of the microtubule:
\begin{equation}\label{bc}
\rho_{0}=\alpha~~\text{and}~~ \rho_{M}=1-\beta.
\end{equation} 
We note here that it is assumed that boundary attachment rates have a corresponding relationship to those inside the microtubule, that is, 
\begin{equation}\label{simple_alpha_and_beta}
\alpha = \dfrac{\omega_A}{\omega_{A}+\omega_{D}} \text{ and }
\beta = \dfrac{\omega_D}{\omega_{A}+\omega_D},
\end{equation} 
and $v_i\equiv v$, then the solution of \eqref{original_system_stationary}-\eqref{bc} is simply a constant: $\rho\equiv \omega_A/(\omega_A+\omega_D)$. However, in practice, rates $\alpha$ and $\beta$ are different from \eqref{simple_alpha_and_beta} which leads to non-trivial stationary solutions possessing interior jumps, even in the homogeneous case $v_i\equiv \text{const}$ \cite{ParFraFre2003,ParFraFre2004}. 

We end this subsection with another form of \eqref{original_system_stationary} which is helpful for understanding the continuous limit presented in Section~\ref{sec:asymp_behaivor}. Introduce the following notation for the flux between sites $x_{i-1}$ and $x_i$:
\begin{equation}\label{discrete_flux}
J_{i-\frac{1}{2}}=v_{i-\frac{1}{2}}\rho_{i-1}(1-\rho_i).
\end{equation}
Note that if one interpolates $\rho_i$ by a smooth function $\rho(x)$ such that $\rho(x_i)=\rho_i$, then performing Taylor expansions one can verify that 
\begin{equation}\label{semi-cont}
\dfrac{\Delta x}{2}\rho'_{i-\frac{1}{2}}= -v_{i-\frac{1}{2}}^{-1}J_{i-\frac{1}{2}}+\rho_{i-\frac{1}{2}}(1-\rho_{i-\frac{1}{2}})+o(\Delta x),
\end{equation} 
where $\rho_{i-\frac{1}{2}}=\rho(x_{i-\frac{1}{2}})$, $\Delta x=x_{i}-x_{i-1}$, and $x_{i-\frac{1}{2}}=x_{i}-\frac{1}{2}\Delta x$.

Going back to definition \eqref{discrete_flux}, we point out that, in terms of $J_{i\pm\frac{1}{2}}$, equation \eqref{original_system_stationary} has the following form: 
\begin{equation}\label{eq_in_terms_fluxes}
J_{i+\frac{1}{2}}-J_{i-\frac{1}{2}}=\omega_{A}  (1-\rho_i)-\omega_D  \rho_i
\end{equation}


\subsection{Limiting continuous problem}
\label{sec:asymp_behaivor}
We focus on the asymptotic behavior of solutions of \eqref{original_system_stationary}-\eqref{bc} as $M\to \infty$ in the framework of the mean-field limit. Specifically, we introduce parameter $\varepsilon:=\ell M^{-1}\ll 1$ (here $\ell$ represents the total length of microtubule) and lattice $\left\{x_i=i\varepsilon,\, i=0,..,M\right\}$ with the distance between lattice points $\ve$. For small $\ve$, we approximate the solution with the continuous function $\rho_\ve(x)$ defined on $0<x<\ell$ and derived from a discrete set of unknowns $\rho_i$ associated with  the lattice points $x_i$, {\it i.e.,} $\rho_\ve(x_i)=\rho_i$.  Then for $\ve\ll 1$, the system of algebraic equations \eqref{original_system_stationary} for the unknown $\rho_i$'s becomes a second order ordinary differential equation for unknown function $\rho_\ve(x)$:
\begin{equation}\label{mean_field_original}
\partial_x\left(v(x)\left(-\dfrac{\varepsilon}{2}\partial_x\rho_{\ve}+\rho_{\ve}(1-\rho_{\ve})\right)\right)=\Omega_A-(\Omega_A+\Omega_D)\rho_\ve,
\end{equation} 
where $v(x)$ is the velocity the motor proteins move with at location $x$, and $\Omega_{A/D}=M\omega_{A/D}$ are properly rescaled attachment/detachment rates. The equalities in \eqref{bc} become boundary conditions for $\rho_\ve(x)$:
\begin{equation}\label{ode_bc}
\rho_\ve(0)=\alpha ~~\text{and}~~ \rho_\ve(\ell)=1-\beta. 
\end{equation} 

\begin{remark}
In order to obtain \eqref{mean_field_original} from \eqref{original_system_stationary} we take the discrete-to-continuous limit $\varepsilon \to 0$. Specifically, we write both \eqref{semi-cont} and \eqref{eq_in_terms_fluxes}, which considered together are equivalent to \eqref{original_system_stationary},  with $\Delta x=\varepsilon\ll 1$:
\begin{equation}
\left\{
\begin{array}{l}
\dfrac{\varepsilon }{2}\rho'_{i-\frac{1}{2}}= -v_{i-\frac{1}{2}}^{-1}J_{i-\frac{1}{2}}+\rho_{i-\frac{1}{2}}(1-\rho_{i-\frac{1}{2}})+o(\varepsilon ),\\
J_{i}'=\Omega_{A}  (1-\rho_i)-\Omega_D  \rho_i +o(\varepsilon ).
\end{array}
\right.
\end{equation}  
We then disregard the $o(\varepsilon)$ terms and write resulting equations for all $x\in(0,\ell)$: 
\begin{equation}\label{key_system}
\left\{
\begin{array}{l}
\dfrac{\varepsilon }{2}\rho'= -v^{-1}J+\rho(1-\rho),\\
J'=\Omega_{A}  (1-\rho)-\Omega_D  \rho\left(=\Omega_A-(\Omega_A+\Omega_D)\rho\right).
\end{array}
\right.
\end{equation}  
Finally, we find $J$ from the first equation of the system above and substitute it into the second equation to derive \eqref{mean_field_original} for $\rho=\rho_\varepsilon$. Here both notations $\rho'$ and $\partial_x\rho$ denote the derivative in $x$. It turns out that the phase portrait of system \eqref{key_system} of two coupled first order differential equations is the key to constructing the solutions of \eqref{mean_field_original}-\eqref{ode_bc}, see Appendix \ref{appendix:homog}.  
\end{remark}

\begin{remark}
Alternatively, in the case of the continuous velocity $v(x)$, this equation can be formally derived from \eqref{original_system_stationary} by using a Taylor expansion of a smooth function $\rho_\ve(x)$ which is again obtained by interpolation of the values $\rho_i$ on the mesh points $x_i$. If $v(x)$ is piecewise continuous with a finite number of jumps, then one needs to supplement equation \eqref{mean_field_original}, which holds inside the intervals of continuity of $v(x)$, with a continuity condition for the flux 
\begin{equation*}
J_\ve(x):=v(x)\left(-\dfrac{\varepsilon}{2}\rho_{\ve}'+\rho_{\ve}(1-\rho_{\ve})\right).
\end{equation*}    
\end{remark}

One can also consider \eqref{mean_field_original} in the distributional sense. Then using the definition of $J_\ve(x)$, the integration of \eqref{mean_field_original} in $x$, and the absolute continuity of integrals together imply the continuity of flux $J_\ve(x)$: 
\begin{eqnarray}
J_\ve(x)&=&J_\ve(0)+\int\limits_0^{x} \left[\Omega_A- (\Omega_A+\Omega_D) \rho_\ve (s)\right] \,\text{d}s\nonumber \\
&=&J_\ve(0)+\Omega_A x - (\Omega_A+\Omega_D)\int\limits_0^x \rho_\ve (s) \,\text{d}s.\label{continuiuty_of_fluxes}
\end{eqnarray}

 In what follows below, we assume that for all $\ve >0$, there exists unique smooth solution $0\leq \rho_\ve(x)<1$ which solves \eqref{mean_field_original} and satisfies the boundary conditions \eqref{ode_bc}. Moreover, there exists a piecewise smooth function $\rho_0(x)$, referred to as the ``outer solution", such that
\begin{equation}\label{def_of_outer_solution}
\lim\limits_{\ve \to 0} \rho_\ve (x)= \rho_0(x), ~~\text{for all }0\leq x \leq \ell.
\end{equation}

By calling $\rho_0(x)$ the outer solution we stick to the standard terminology of singularly perturbed ordinary differential equations \cite{Hol2013}. While the outer solution $\rho_0(x)$ is the pointwise limit of $\rho_\ve(x)$ for $0\leq x \leq \ell$, it does not approximate the exact solution $\rho_\ve (x)$ uniformly on $0\leq x \leq \ell$. Specifically, the outer solution $\rho_0(x)$ approximates $\rho_\ve(x)$ poorly in the vicinity of the jumps $\left\{x_J\right\}$.  To make the approximation uniform, one takes into account boundary layer terms of the form $\rho_{\text{corrector}}(x)=Y((x-x_J)/\ve)$ whose distinguishing feature is that its slope, derivative in $x$, is of the order of $\ve^{-1}$. Observe also that, even though $\rho_0(0)=\alpha$ and $\rho_0(\ell)=1-\beta$, it is possible that the outer solution $\rho_0(x)$ does not satisfy the boundary condition in the following sense: either $\lim_{x\to 0^+} \rho_0(x)\neq \alpha$ or $\lim_{x\to \ell^-}\rho_0(x)\neq 1-\beta$.  

\begin{remark}
We note that \eqref{original_system} possesses a unique constant solution $\rho_i\equiv \Omega_A/(\Omega_A+\Omega_D)$.  In what follows, we are interested in regimes where jamming can occur and, thus, for the sake of simplicity we restrict ourselves to the case when the attachment rate exceeds the detachment rate, that is, $\Omega_A>\Omega_D$. This implies that the constant solution of \eqref{original_system} is  greater than $1/2$. 
\end{remark}
        




\subsection{Homogeneous microtubule}
\label{sec:homogeneous}

Consider a constant motor protein motility rate $v(x)$, representing a homogeneous microtubule, that is, $v_i\equiv v_0$. The solution in the homogeneous case has been studied previously (e.g., \cite{ParFraFre2003,PopRakWilKolSch2003,ParFraFre2004}). In this case, \eqref{mean_field_original} reduces to
\begin{eqnarray}
&&v_0\partial_x\left(-\dfrac{\varepsilon}{2}\partial_x\rho_\ve+\rho_{\ve}(1-\rho_{\ve})\right)=\Omega_A-(\Omega_A+\Omega_D)\rho_\ve,\quad 0<x<\ell,\label{homo_mean_field}\\
&&\rho_\ve(0)=\alpha,\quad \rho_\ve(\ell)=1-\beta. \label{homo_bc}
\end{eqnarray} 
Recall that $\ell$ is the length of a microtubule. 
 
Equation \eqref{homo_mean_field} is a second order nonlinear ODE. If one formally passes to the limit $\ve \to 0$ in \eqref{homo_mean_field}, then the term with the second derivative of $\rho_\ve$ vanishes and this equation  becomes first order where the solution cannot, in general, satisfy both boundary conditions in \eqref{homo_bc} as the boundary value problem is overdetermined. To describe the limiting solution, $\lim\limits_{\ve \to 0}\rho_\ve(x)$,  we introduce an auxiliary function $g(x;s,a)$, which is the solution of the initial value problem of the first order obtained from the formal limit as $\ve \to 0$ in \eqref{homo_mean_field}: 
\begin{equation}
\label{def_of_gg}
\left\{
\begin{array}{l}
v_0(1-2g)\partial_x g=\Omega_A-(\Omega_A+\Omega_D)g,\\ g(x=s;s,a)=a.
\end{array}\right.
\end{equation} 
In other words, function $\rho_0(x):=g(x;s,a)$ is the smooth outer solution subject to a single boundary condition (or, equivalently, initial condition): $\rho_0|_{x=s}=a$. Equation \eqref{def_of_gg} can be solved in an explicit form in terms of special functions (see~\cite{ParFraFre2004}). 

Note that initial value problem \eqref{def_of_gg} is not well-posed for $a=0.5$. If $a=0.5$, then there is no solution for $x>s$, and for $x<s$ we define $g$ as the function solving the differential equation in \eqref{def_of_gg} subject to the following condition: 
\begin{equation}
\lim\limits_{x\to s}g(x;s,0.5)=0.5 \text{ and } g(x;s,0.5)>0.5 \quad \text{for all }x< s. 
\end{equation}  
In other words, for the initial condition with $a=0.5$, equation \eqref{def_of_gg} admits two solutions for $x<s$: 
one solution is less than $0.5$, another one is above $0.5$, and to describe the outer solution $\rho_0(x)$ we will need restrict our consideration to the upper one (the upper solution is stable in a certain sense, see Appendix~\ref{appendix:homog}). 

Function $g(x;s,a)$ is not necessarily defined globally, for all $x$. For example, consider $x\geq s$, then the solution exists on the interval $(s,s+x_a)$ for some $x_a>0$ and at $x=s+x_a$ the slope of $g$ becomes unbounded. For example, if $a<0.5$, then the value of $x_a$ can be found from the condition $g(s+x_a;s,a)=0.5$ which can be written as  
\begin{equation}\label{def_of_x_a}
x_a=\int_{a}^{0.5}\dfrac{\text{d}g}{\boldsymbol{v}(g)}=\int_{a}^{0.5}\dfrac{v_0(1-2g)\,\text{d}g}{\Omega_A-(\Omega_A+\Omega_D)g},
\end{equation}
where function $\boldsymbol{v}(g)$ is introduced in such a way that the differential equation from \eqref{def_of_gg} is equivalent to $\partial_x g =\boldsymbol{v}(g)$. One can compute the integral on the right-hand side of \eqref{def_of_x_a} to obtain an analytic formula for $x_a$: 
\begin{equation}
x_a=v_0\dfrac{\Omega_A-\Omega_D}{(\Omega_A+\Omega_D)^2}\log\dfrac{\Omega_A-\Omega_D}{2(\Omega_A-a (\Omega_A+\Omega_D))}+v_0\dfrac{1-2a}{\Omega_A+\Omega_D}.
\end{equation}

The following theorem gives an explicit formula for the limiting solution of the \eqref{homo_mean_field}-\eqref{homo_bc} as $\ve \to 0$. 
\begin{theorem}\label{thm:homog}
	Define $\rho_0(x):=\lim\limits_{\ve \to 0} \rho_\ve(x)$ for $0\leq x \leq \ell$ and $\rho_{\ve}$ solving \eqref{homo_mean_field}-\eqref{homo_bc}. 
	Then
	\begin{equation}\label{construction_of_outer_solution}
	\rho_0(x)=
	\left\{
	\begin{array}{ll} 
	\alpha, & x=0,\\
	g(x;0,\alpha), & 0<x<\max\{0,x_J\},\\
	g(x;\ell,\max\{0.5,1-\beta\}), & \max\{0,x_J\}<x <\ell,\\
	1-\beta, & x=\ell.
	\end{array} 
	\right.
	\end{equation}
If $\alpha\geq 1/2$, then $x_{J}=0$. If $\alpha < 1/2$, then $x_J$ is determined by 
	\begin{equation}\label{def_of_x_J}
	x_J:=\mathrm{min}\left\{x\geq 0\,|\,\, g(x;0,\alpha)+g(x;\ell,\max\left\{0.5,1-\beta\right\})\leq 1\right\}.
	\end{equation}
\end{theorem}
 
The result of this theorem is consistent with previous works where the system \eqref{homo_mean_field}-\eqref{homo_bc} was studied (e.g., see \cite{Fre11,Fre07,Fre18}). We relegate the proof of this theorem  and examples of the application of the representation formula \eqref{construction_of_outer_solution} to Appendix~A.  
 
\begin{remark}
The point $x_J$, if $0<x_J<\ell$, is the location of the interior jump, that is, the outer solution $\rho_0(x)$ is a smooth solution of the differential equation from \eqref{def_of_gg} on intervals $(0,x_J)$ and $(x_J,\ell)$ and it has one jump inside $(0,\ell)$ at $x=x_J$. 
The value of $x_J$ can also be found via numerical simulations of the following equation 
	\begin{equation}\label{def_of_x_J_2}
	g(x_J;0,\alpha) + g(x_J;\ell,\max\{0.5,1-\beta\})=1.
	\end{equation}
This equation is equivalent to the continuity of fluxes $J_0=\rho_0(1-\rho_0)$ at the point of the jump of the outer solution $\rho_0=\rho_0(x)$: 
\begin{equation}
v_0\rho_{\alpha}(1-\rho_\alpha)|_{x\to x_J-}   =v_0\rho_{\beta}(1-\rho_\beta)|_{x\to x_J+},  	
\end{equation}
where $\rho_\alpha(x)=	g(x;0,\alpha)$ and $\rho_\beta(x)=g(x;\ell,\max\{0.5,1-\beta\})$.  
\end{remark}

\smallskip 

\begin{remark}
If one varies boundary conditions \eqref{homo_bc}, then the outer solution $\rho_0(x)$ may stay unchanged (except values at $x=0$ and $x=\ell$) for wide range of parameters $\alpha$ and $\beta$. For example, denote $\rho_\beta(x):=g(x,\ell,\max\left\{0.5,1-\beta\right\})$. Theorem~\ref{thm:homog} implies that outer solution $\rho_0(x)$ coincides with $\rho_\beta(x)$ on the interval $0<x<1$ for the following range of $\alpha$: 
\begin{equation}
1-\rho_{\beta}(0)\leq \alpha \leq 1.
\end{equation}
Once $\alpha$ becomes smaller than the lower limit, $1-\rho_{\beta}(0)$, an interior jump appears in the outer solution $\rho_0(x)$.       
\end{remark}
\smallskip

The following corollary is important for the study of inhomogeneous microtubules in Section~\ref{sec:nonhomo}.   
\begin{corollary}\label{thm2:homog}
Assume  $\alpha\geq 1/2$ or
\begin{equation}\label{no_gamma_l}
\int\limits_{0}^{0.5}\dfrac{v_0(1-2\rho)}{\Omega_A-(\Omega_A+\Omega_D)\rho} \,\mathrm{d}\rho \leq \ell.
\end{equation}
Then
\begin{itemize}
\item[(i)] $\lim\limits_{x\to \ell-}\rho_0(x)=\max\left\{0.5,1-\beta\right\}$.
\item[(ii)] If $\beta<1/2$, then $\rho_0(x)$ is continuous at $x=\ell$, that is, $\lim\limits_{x\to\ell-}\rho_0(x)=\rho_0(0)=1-\beta$. 
\item[(iii)] If $\alpha \geq 1/2$, then there is no interior jump, that is, $x_J\leq 0$.
\item[(iv)] If $0<x_J<\ell$, then $\alpha<1/2$ and $\lim\limits_{x\to 0+}\rho_0(x)=\alpha$ (that is, $\rho_0(x)$ is continuous at $x=0$).
\end{itemize}	
\end{corollary}

\smallskip 

Condition \eqref{no_gamma_l} excludes the case of low density solutions, that is, we exclude outer solutions of the form: $\rho_0(x)<0.5$ for all $x\in [0,\ell)$. This regime is not consistent with jamming, which is the focus of this work.  In other words, condition \eqref{no_gamma_l} imposes that the ``left" part of solution, $g(x;0,\alpha)$, cannot be extended to entire interval $[0,\ell)$. The reason we would like to impose this condition below in the inhomogeneous case is because we then focus on cases when regions with high densities emerge, and thus traffic jams in motor transport are possible. By direct integration, condition \eqref{no_gamma_l} can be written as 
\begin{equation}
\dfrac{v_0}{\Omega_A+\Omega_D}\left[1-\dfrac{\Omega_A-\Omega_D}{\Omega_A+\Omega_D}\log\left(\dfrac{2}{\Omega_A-\Omega_D}\right)\right]\leq \ell.
\end{equation} 


\subsection{Inhomogeneous microtubule} \label{sec:nonhomo}
In this section, we consider a non-con\-stant motility rate $v(x)$. Biologically, it corresponds to a inhomogeneous microtubule with different motor protein mobilities in different regions of the microtubule. Without loss of generality, we take $\ell = 1$. We restrict ourselves to the case
\begin{equation}\label{def_of_two_valued_v}
v(x)=\left\{\begin{array}{ll}v_L,&0\leq x\leq x_0,\\v_R,&x_0< x\leq 1.\end{array}\right.
\end{equation}
This case may be considered as two coupled homogeneous microtubules meeting at  interface $x=x_0$ with coupling through the continuity of densities and fluxes. Specifically, we have the following system of equations: 
\begin{eqnarray}
&& v_L\partial_x\left[\left(-\dfrac{\varepsilon}{2}\rho_{\ve}'+\rho_{\ve}(1-\rho_{\ve})\right)\right]=\Omega_A-(\Omega_A+\Omega_D)\rho_\ve,~~~0<x<x_0,\label{pde_left}\\
&& v_R\partial_x\left[\left(-\dfrac{\varepsilon}{2}\rho_{\ve}'+\rho_{\ve}(1-\rho_{\ve})\right)\right]=\Omega_A-(\Omega_A+\Omega_D)\rho_\ve,~~~x_0<x<1,\label{pde_right}
\end{eqnarray} 
and two coupling conditions: 
\begin{enumerate}
	\item[(1)] continuity of $\rho_\ve(x)$ at $x_0$: 
	\begin{equation}
	\label{continuity_of_rho_ve}
	\rho_\ve(x_0^-)=\rho_\ve(x_0^+)=:\rho_\ve.
	\end{equation}
	\item[(2)] continuity of flux $J_\ve(x)$ at $x_0$:
\begin{equation}\label{cont_fluxes_eps}
v_L \left(-\dfrac{\varepsilon}{2}\rho_{\ve}'(x_0^-)+\rho_{\ve}(1-\rho_{\ve})\right)=v_R \left(-\dfrac{\varepsilon}{2}\rho_{\ve}'(x_0^+)+\rho_{\ve}(1-\rho_{\ve})\right).
\end{equation}
\end{enumerate}
 
 The outer solution $\rho_0(x)=\lim\limits_{\ve\to0}\rho_\ve(x)$ is not necessarily continuous, nevertheless due to \eqref{continuiuty_of_fluxes} it satisfies the flux continuity condition:
 \begin{equation}\label{limiting_continuity_of_fluxes}
\displaystyle  v_L \rho_0(1-\rho_0)\biggl|_{x\to x_0^-}=v_R \rho_0(1-\rho_0)\biggr|_{x\to x_0^+}.
 \end{equation} 
 
\smallskip 

\begin{remark} By looking at the system \eqref{pde_left}-\eqref{pde_right} one may think that the case of inhomogeneous $v(x)$ is equivalent to the case of inhomogeneous attachment/detach\-ment rates but with constant motility rates $v(x)\equiv 1$:
\begin{equation}\label{nonhom_omegas}
\partial_x\left(-\dfrac{\varepsilon}{2}\rho_{\ve}'+\rho_{\ve}(1-\rho_{\ve})\right)=\Omega_A(x)-(\Omega_A(x)+\Omega_D(x))\rho_\ve
\end{equation}
with attachment/detachment rates $\Omega_A(x)$ and $\Omega_D(x)$ are $\Omega_A/v_L$ and $\Omega_D/v_L$ inside the left half of the microtubules, $x<x_0$, and $\Omega_A/v_R$ and $\Omega_D/v_R$ inside the right half, $x>x_0$. For differential equation \eqref{nonhom_omegas}, from the continuity of $\rho_\ve(x)$ and $J_\ve(x)$, one concludes that $\rho_\ve$ is necessarily continuously differentiable, whereas the solution of \eqref{pde_left}-\eqref{pde_right} possesses a discontinuous derivative at $x_0$ for $v_L\neq v_R$ which follows from \eqref{cont_fluxes_eps} and is written as 
\begin{equation}
\rho'_\ve(x_0^+)-\rho'_\ve(x_0^-)=\dfrac{2}{\ve v_Rv_L (v_R-v_L)} J_{\ve}(x_0).
\end{equation}
 Therefore, problems for inhomogeneous $v(x)$ and inhomogeneous $\Omega_{A,D}(x)$ are not equivalent.  
\end{remark}

\medskip 

By analogy with function $g$ from \eqref{def_of_gg} in the case of homogeneous microtubules, we introduce $g_L(x;s_L,a_L)$ and $g_R(x;s_R,a_R)$ as solutions of the following initial value problems: 
\begin{equation}
\label{def_of_g_LR}
\left\{\begin{array}{l}v_L(1-2g_L)\partial_x g_L=\Omega_A-(\Omega_A+\Omega_D)g_L,~~~g_L(s_L;s_L,a_L)=a_L,\\
v_R(1-2g_R)\partial_x g_R=\Omega_A-(\Omega_A+\Omega_D)g_R,~~~g_R(s_R;s_R,a_R)=a_R.\end{array}\right.
\end{equation} 

In what follows we consider separately two cases:
\begin{itemize}
\item[(i)] {\it fast - slow microtubule}: $v_L>v_R$,
\item[(ii)] {\it slow - fast microtubule}: $v_R>v_L$.
\end{itemize}
 
Before we formulate our main result for these two cases, we note that the difficulty in the determination of the outer solution $\rho_0(x)$ is finding  the value of $\rho_0$ at the interface,  $A:=\rho_0(x_0)$. Once $A$ is found, one can use Theorem~\ref{thm:homog} to restore $\rho_0(x)$ in both intervals $[0,x_0]$ and $[x_0,1]$. 

\smallskip 

\begin{theorem}\label{thm:nonhomog1}
	Consider $v_L>v_R$ and assume that condition \eqref{no_gamma_l} holds with $v_0=v_L$ and $\ell=x_0$. 
	Let $\rho_0(x)$ be the outer solution of system \eqref{pde_left}-\eqref{pde_right} equipped with coupling conditions \eqref{continuity_of_rho_ve}-\eqref{cont_fluxes_eps}. Then function $\rho_0(x)$ has a jump at $x=x_0$ and is given by 
	\begin{equation}\label{repr_formula_1}
	\rho_0(x)=\left\{
	\begin{array}{ll}
	\alpha,& x=0,\\
	g_L(x;0,\alpha),& 0<x<\max\{0,x_J\},\\
	g_L(x;x_0,A),&\max\{0,x_J\}< x\leq x_0,\\
	g_R(x;1,\max\{0.5,1-\beta\}),&x_0<x<1,\\
	1-\beta,&x=1,
	\end{array}
	\right.
	\end{equation}
	where $x_J$ is determined from the continuity of fluxes:
	\begin{eqnarray}
	g_L(x_J;0,\alpha)+g_L(x_J;x_0,A)=1,
	\end{eqnarray}
	and
	$
	A=\left(v_L+\sqrt{v_L^2-4J^2_R}\right)/(2v_L)
	$ 
	where 
	$
	J_R=v_Rg_R(1-g_R)$  with $$g_R:=g_R(x_0;1,\max\{0.5,1-\beta\}).$$
	\end{theorem}
\begin{proof}
First, we show that 
{the outer solution $\rho_0(x)$ has a jump at $x=x_0$}. Indeed, from continuity of fluxes \eqref{cont_fluxes_eps} with $\rho_\ve(x_0)=A+o(1)$ we get:
 \begin{equation}
 v_L\rho_{\ve}'(x_0^-)=v_R \rho_{\ve}'(x_0^+) + \dfrac{v_L-v_R}{2\ve}A(1-A)+o\left(\dfrac{1}{\ve}\right).
 \end{equation}
Since $A$ is strictly between 0 and 1, we find that one of the derivatives (either the left or right one) is of the order $\ve^{-1}$  corresponding to a jump.

Next, denote the limits of the outer solution from the left and right at $x_0$ by 
\begin{equation*} 
A_L:=\rho_0|_{x\to x_0^-}\text{ and }A_R:=\rho_0|_{x\to x_0^+}.
\end{equation*} 
Then the continuity of fluxes \eqref{limiting_continuity_of_fluxes} is written as $v_LA_L(1-A_L)=v_R A_R(1-A_R)$, and since $v_L>v_R$ we have $A_L(1-A_L)<1/4$. By applying Corollary~\ref{thm2:homog} (ii) we find $A=A_L>0.5$ and  thus, there is no jump from the left at $x_0$. Then, following Corollary~\ref{thm2:homog} (iii), there is no interior jump, and $\rho_0(x)$ in interval $[x_0,1]$ is determined by Corollary~\ref{thm2:homog} (i), which is a boundary condition for $\rho_0(x)$ at $x_0=1$. 

Specifically, by Theorem~\ref{thm:homog} we have $\rho_0(x)=g_R(x;1,\max\{0.5,1-\beta\})$ for $x_0<x<1$.
Then $A_R=g_R:=g_R(x_0;1,\max\{0.5,1-\beta\})$
  and $A$ is the solution of the quadratic equation $v_LA(1-A)=v_RA_R(1-A_R)$, which is strictly greater than $0.5$.  Thus, we found $A$, and the expression for $\rho_0(x)$ in $[0,x_0]$ is found by using the representation formula \eqref{construction_of_outer_solution} from Theorem~\ref{thm:homog}.   
  \qed
\end{proof}

\begin{figure}
	\begin{center}
		\includegraphics[width=0.475\textwidth]{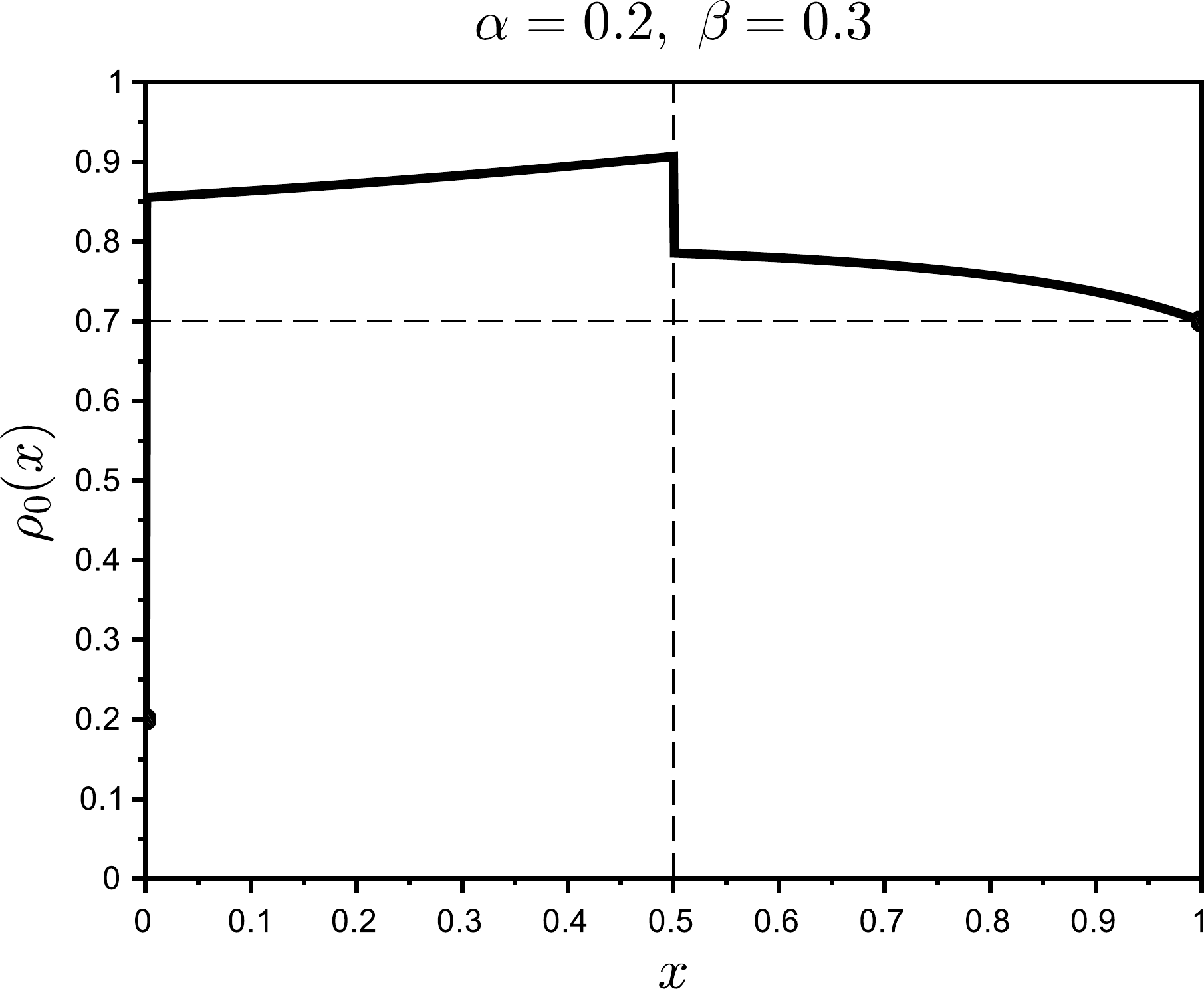}
		\includegraphics[width=0.475\textwidth]{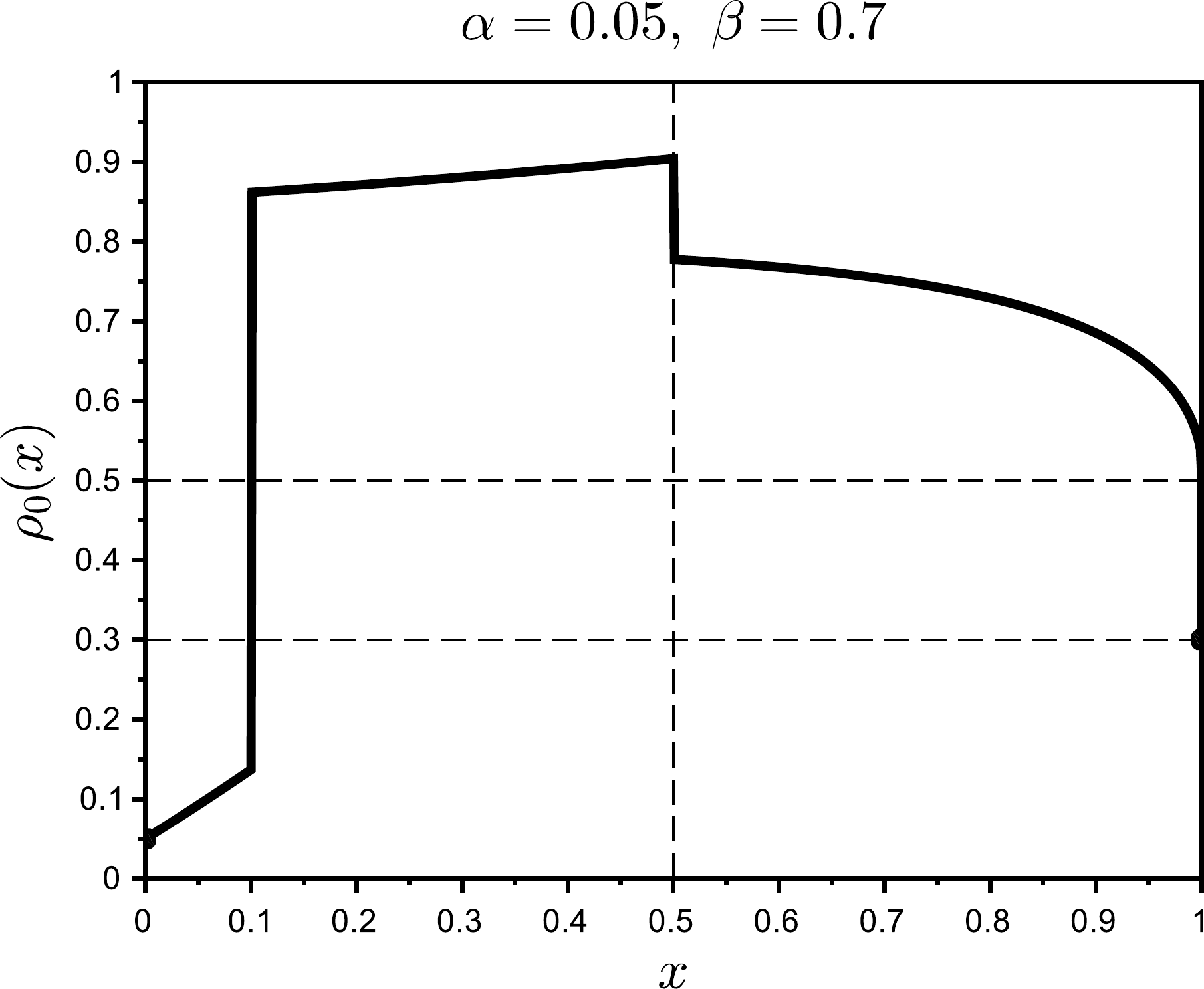}\\
			\includegraphics[width=0.475\textwidth]{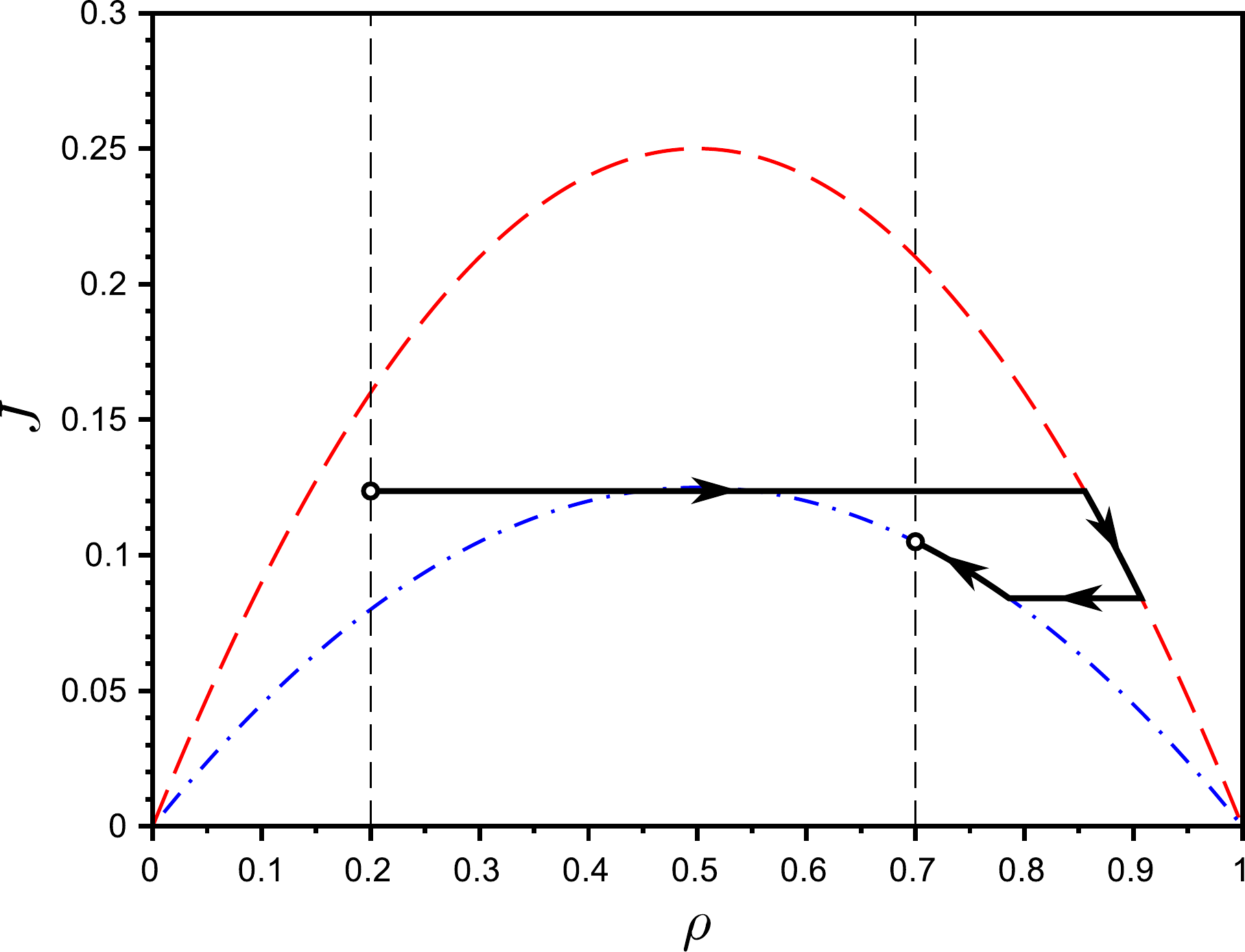}
		\includegraphics[width=0.475\textwidth]{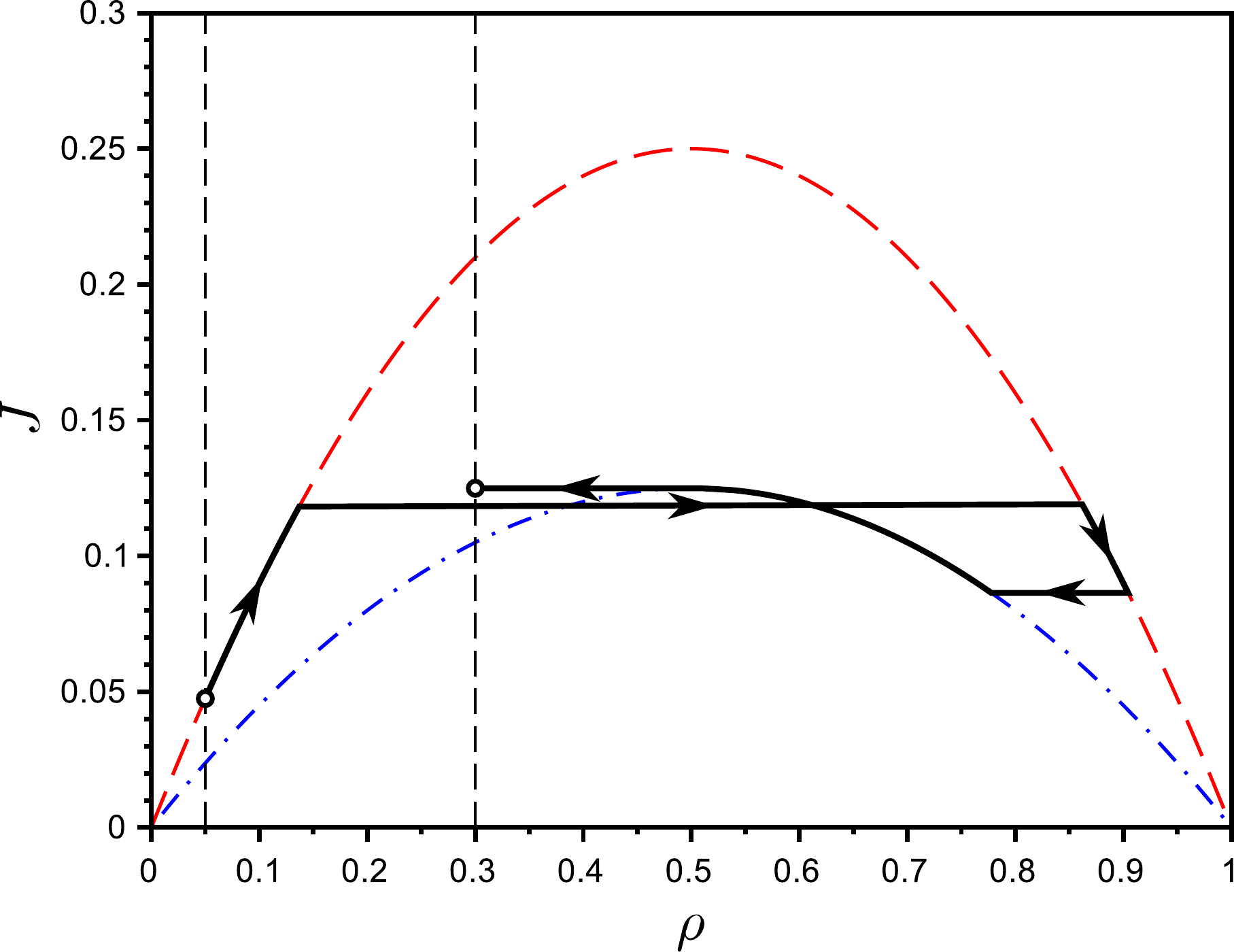}
		\caption{Examples 1 \& 2 for inhomogeneous microtubulus. The upper row depicts plots of outer solutions $\rho_0(x)$. The lower row depicts solutions as trajectories in the $(\rho,J)$-plane; note that these trajectories are continuous in $J$ and have discontinuities (jumps) in $\rho$. Red (lower) and blue (upper) dashed lines are arcs $J=v_R\rho(1-\rho)$ and $J=v_L\rho(1-\rho)$, respectively. }
		\label{fig:nohomog-ex1}
	\end{center}
\end{figure}

\smallskip 

Next, we illustrate formula \eqref{repr_formula_1} of Theorem~\ref{thm:nonhomog1} by the following two examples with $x_0=0.5$, $v_L=1.0$, $v_R=0.5$, $\Omega_A=0.8$ and $\Omega_D=0.2$.  

\smallskip 

 \noindent{\bf Example 1.} $\alpha=0.2$ and $\beta=0.3$.
\begin{equation*}
\rho_0(x)=\left\{
\begin{array}{ll}
0.2,& x=0,\\
g_L(x;0.5,0.907),&0<x\leq 0.5,\, (A=0.5(1\pm \sqrt{1-4J_R})\approx 0.907)\\
g_R(x;1.0,0.7),&0.5<x\leq 1.0.
\end{array}
\right.
\end{equation*}
 
\smallskip

\noindent{\bf Example 2.} $\alpha=0.05$ and $\beta=0.7$. 
\begin{equation*}
\rho_0(x)=\left\{
\begin{array}{ll}
g_L(x;0.0,0.05),& 0\leq x<x_J,\, (x_J\approx 0.101)\\
g_L(x;0.5,0.904),&x_J<x\leq 0.5,\\
g_R(x;1.0,0.5+),&0.5<x<1.0,\, (A\approx 0.904)\\
0.3,&x=1.0.
\end{array}
\right.
\end{equation*}

Solutions from Examples 1 and 2 are depicted in Fig.~\ref{fig:nohomog-ex1}. The main difference between Examples 1 and 2 is that in Example 2 there is an interior jump in the left sub-interval $(0,0.5)$, while the solution in Example 1 does not possesses a jump in neither of the two sub-intervals, left $(0,0.5)$ and right $(0.5,1)$ (or, equivalently, $x_J=0$ in \eqref{repr_formula_1} for Example 1).

\medskip 

\begin{theorem}\label{thm:nonhomog2}
	Consider $v_R>v_L$ and assume condition
	\eqref{no_gamma_l}  holds with both $v_0=v_L$, $\ell=x_0$ and $v_0=v_R$, $\ell=1-x_0$. 
	Let $\rho_0(x)$ be the outer solution of system \eqref{pde_left}-\eqref{pde_right} equipped with coupling conditions \eqref{continuity_of_rho_ve}-\eqref{cont_fluxes_eps}. 
Denote also 
\begin{equation}\label{def_of_hat_ar}
\hat{A}_R:=g_R(x_0;1,\max\{0.5,1-\beta\}).
\end{equation} 	
Then 
\begin{equation}\label{determination_of_A_42}
A=\left\{\begin{array}{ll}\dfrac{1}{2}-\sqrt{\dfrac{1}{4}-\dfrac{v_L}{4v_R}},& \qquad  v_R{\hat{A}_R(1-\hat{A}_R)> 0.25 v_L},\\
&\\
\dfrac{1}{2}+ \sqrt{\dfrac{1}{4}-\dfrac{v_R}{v_L}\hat{A}_R(1-\hat{A}_R)},& \qquad v_R{\hat{A}_R(1-\hat{A}_R)\leq 0.25 v_L},
 \end{array}\right.
\end{equation}
	and function $\rho_0(x)$ is given by
	\begin{equation}\label{formula_from_thm42}
	\rho_0(x)=\left\{
	\begin{array}{ll}
	\alpha,& \qquad x=0,\\
	g_L(x;0,\alpha),& \qquad  0<x<\max\{0,x_J^{(L)}\},\\
	g_L(x;x_0,\max\{A,0.5\}),& \qquad \max\{0,x_J^{(L)}\}\leq x<x_0,\\
	g_R(x;x_0,\min\{A,0.5\}),&\qquad  x_0\leq x\leq x_J^{(R)},\\
	g_R(x;1,\max\{0.5,1-\beta\}),& \qquad  x_J^{(R)}<x<1,\\
	1-\beta,& \qquad x=1.
	\end{array}
	\right.
	\end{equation} 
Here $x_J^{(L)}$ and $x_J^{(R)}$ are determined from the continuity of fluxes:
\begin{eqnarray*}
&g_L(x_J^{(L)};0,\alpha)+g_L(x_J^{(L)};x_0,\max\{A,0.5\})=1,&\\
&g_R(x_J^{(R)};x_0,\min\{A,0.5\})+g_R(x_J^{(R)};1,\max\{0.5,1-\beta\})=1.&
\end{eqnarray*}
\end{theorem}	

\begin{proof}
As in the proof of Theorem~\ref{thm:nonhomog1}, introduce $A_L:=\rho_0|_{x\to x_0^-}$ and $A_R~:=~\rho_0|_{x\to x_0^+}$. Then continuity of fluxes at $x_0$ reads 
\begin{equation}\label{cont_of_fluxes_in_proof}
v_LA_L(1-A_L)=v_RA_R(1-A_R).
\end{equation} 
By Corollary~\ref{thm2:homog} (i) we get $A_L\geq 0.5$. 
Next, we consider two cases. If $$
\dfrac{v_R}{v_L}\hat{A}_R(1-\hat{A}_R)\leq 1/4,
$$
then $A_R=\hat{A}_R$ and $A_L\geq 0.5$ solves \eqref{cont_of_fluxes_in_proof}, that is, 
\begin{equation*}
A=A_L = \dfrac{1}{2}+ \sqrt{\dfrac{1}{4}-\dfrac{v_R}{v_L}\hat{A}_R(1-\hat{A}_R)}.
\end{equation*}
If
$
\dfrac{v_R}{v_L}\hat{A}_R(1-\hat{A}_R)> 1/4,
$
then $\rho_0(x)$ does not coincide with $g_R(x;1,\max\{0.5,1-\beta\})$ on the entire interval $(x_0,1)$, there is an interior jump at $x_0<x_J^{(R)}<1$ and by Corollary~\ref{thm2:homog} (iii) and (iv), $A_R<1/2$ and $A=A_R$. Then $A_L=1/2$ and $\rho_0(x)$ on intervals $(0,x_0)$ and $(x_0,1)$ is found by \eqref{construction_of_outer_solution}. 
\qed 
\end{proof}

\smallskip 

Next we illustrate formula \eqref{formula_from_thm42} by two examples with $v_R>v_L$. Namely, set $v_L=0.5$ and $v_R=1.0$ and consider $x_0=0.5$, $\Omega_A=0.8$ and $\Omega_D=0.2$. 

\begin{figure}[t]
	\begin{center}
		\includegraphics[width=0.475\textwidth]{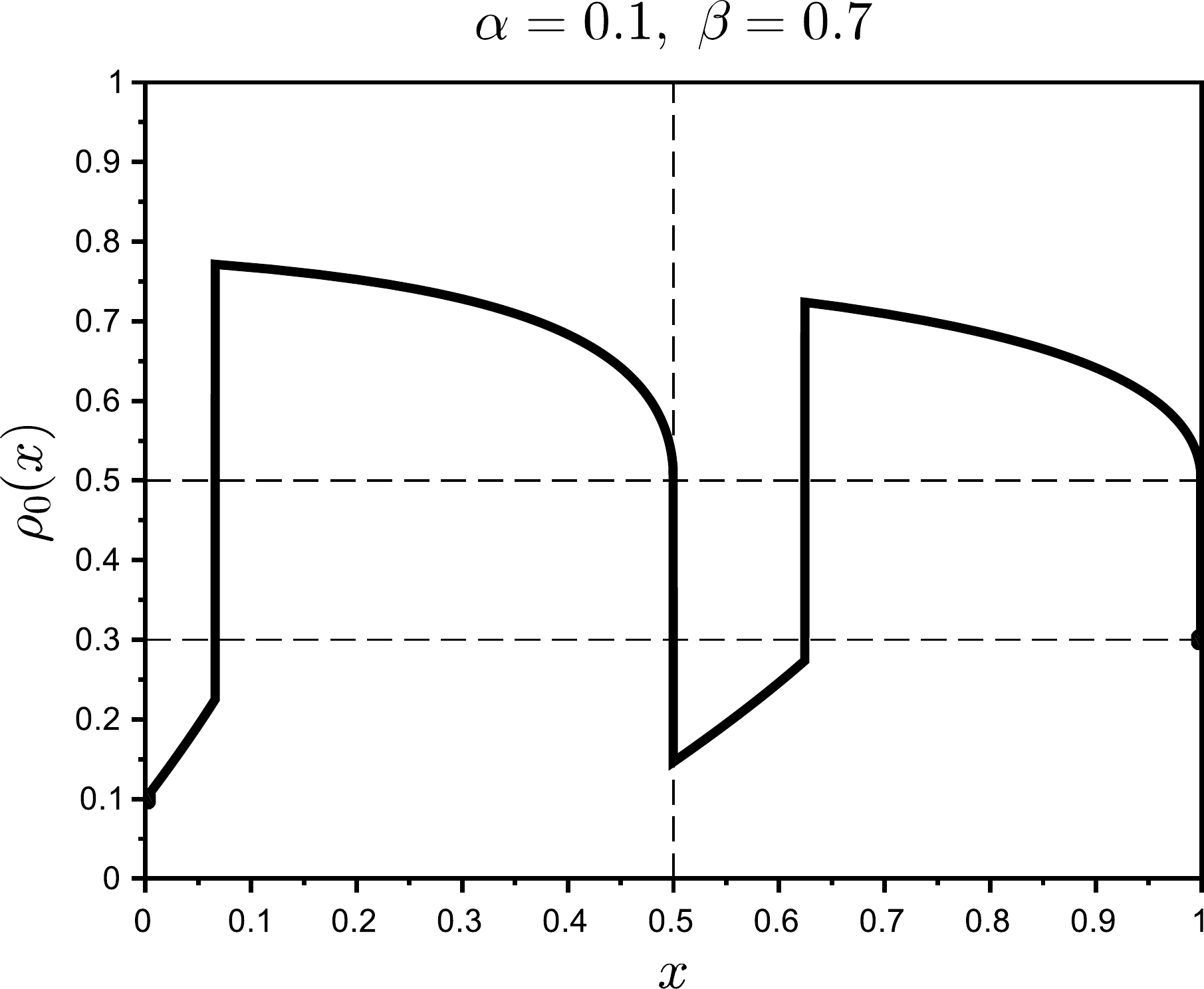}
		\includegraphics[width=0.475\textwidth]{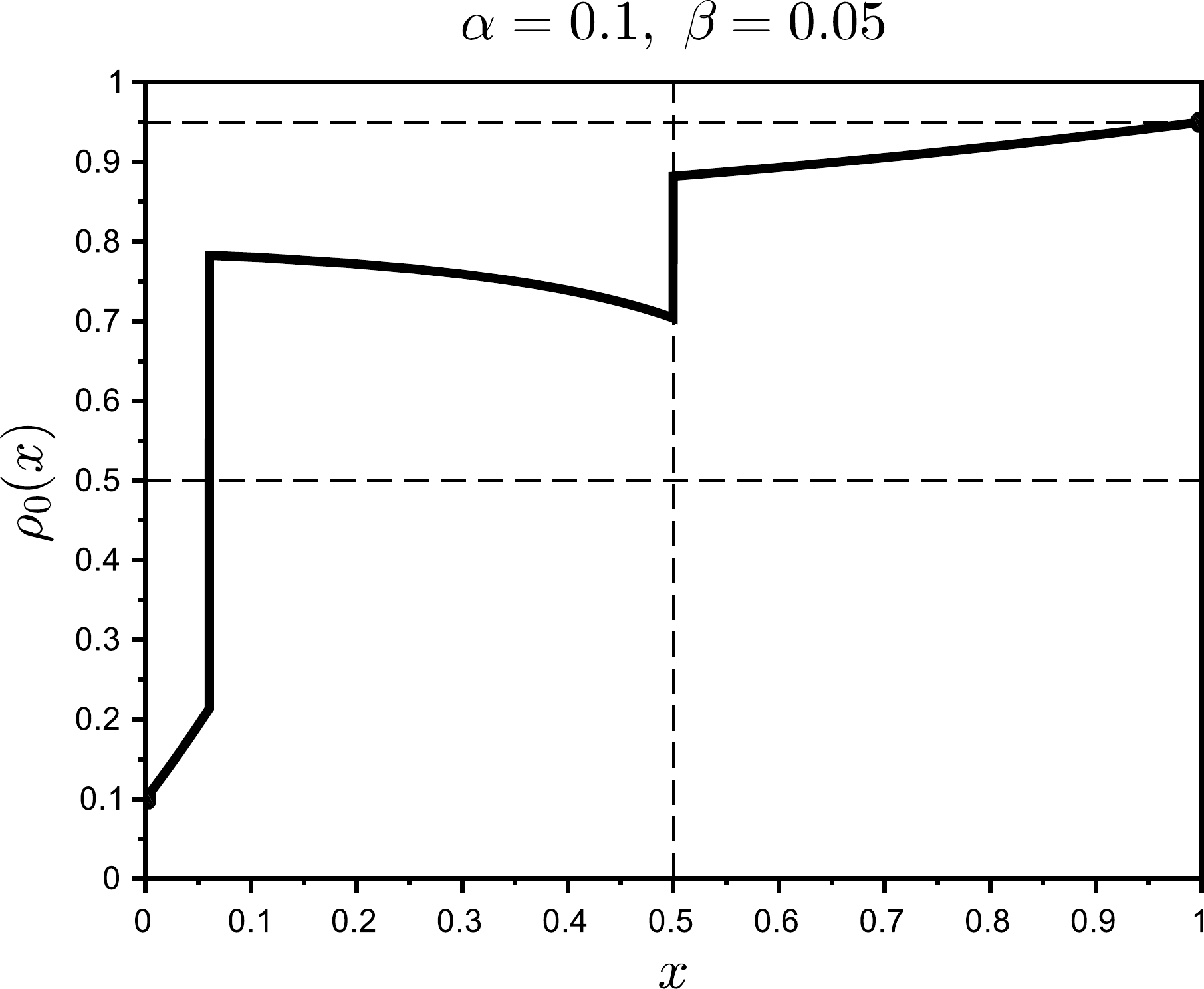}\\
			\includegraphics[width=0.475\textwidth]{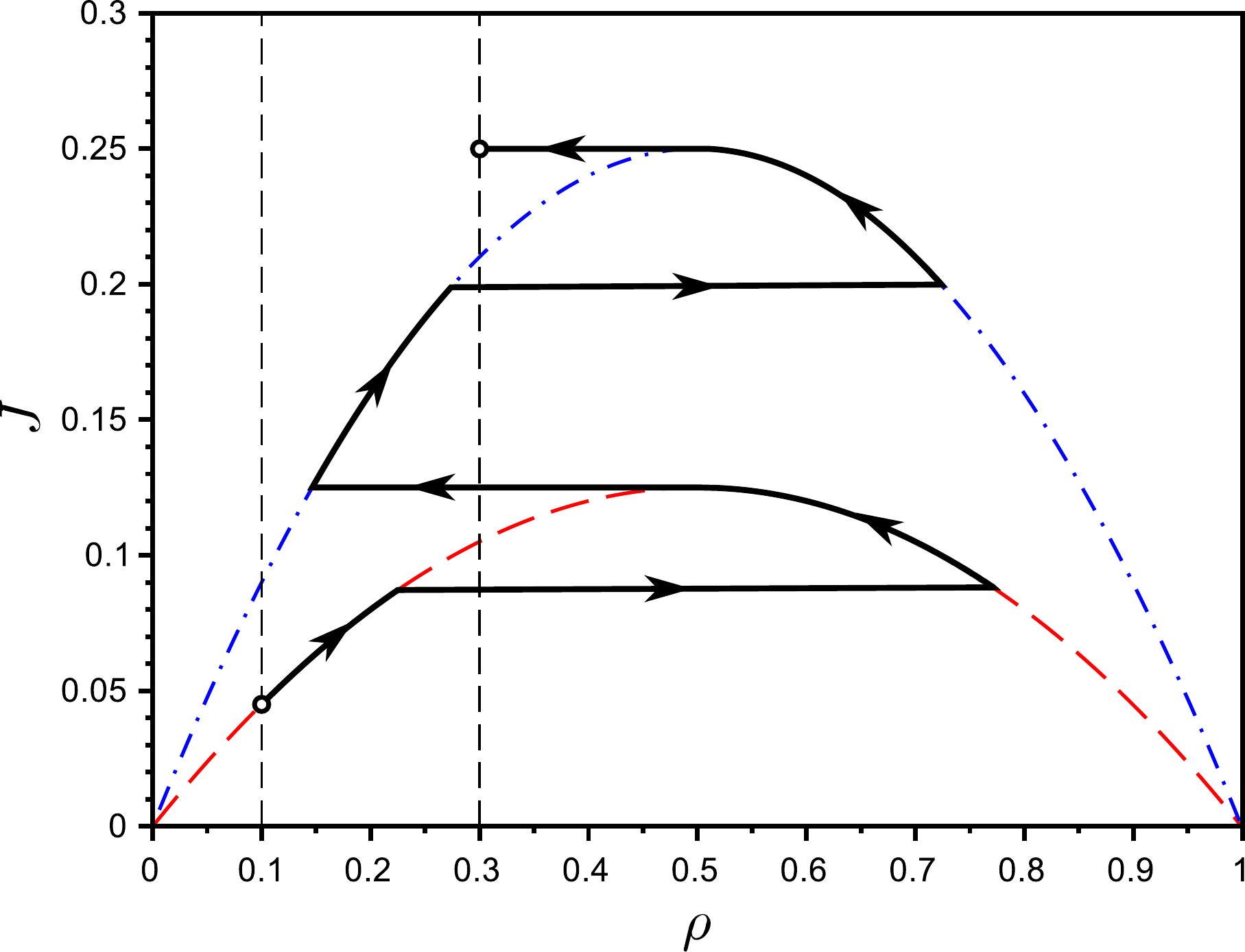}
		\includegraphics[width=0.475\textwidth]{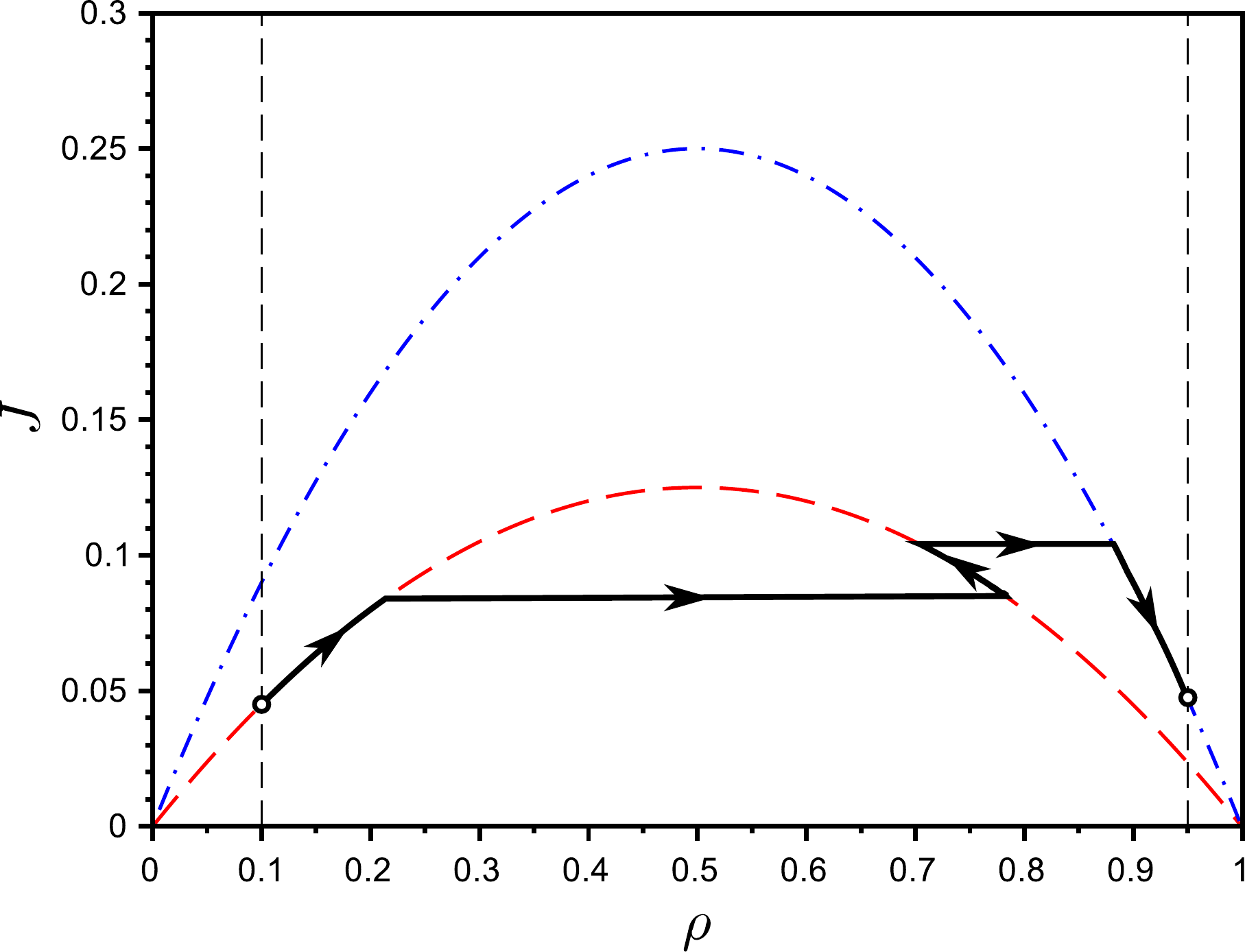}
		\caption{Examples 3 \& 4 for a inhomogeneous microtubule. The upper row depicts plots of outer solutions $\rho_0(x)$. The lower row depicts solutions as trajectories in $(\rho,J)$ plane, red (lower) and blue (upper) dashed lines are arcs $J=v_L\rho(1-\rho)$ and $J=v_R\rho(1-\rho)$, respectively.  }
		\label{fig:nohomog-slow}
	\end{center}
\end{figure}

\medskip 

\noindent{\bf Example 3.} $\alpha=0.1$ and $\beta=0.7$.
\begin{equation*}
\rho_0(x)=\left\{
\begin{array}{ll}
g_L(x;0.0,0.1),& \qquad  0\leq x<x_J^{(L)}, \,(x_J^{(L)}\approx 0.097)\\
g_L(x;0.5,0.5+),& \qquad x_J^{(L)}<x< 0.5,\\
g_R(x;0.5,A),& \qquad  0.5 \leq  x < x_J^{(R)},\, (A=\frac{1}{2}(1-\frac{1}{\sqrt{2}}), \, x_J^{(R)}\approx 0.624)\\
g_R(x;1.0,0.5+),& \qquad  x_J^{(R)}<x<1.0,\\ 
0.3,& \qquad  x=1.0.
\end{array}
\right.
\end{equation*}
\noindent{\bf Example 4.} $\alpha=0.1$ and $\beta=0.05$. 
\begin{equation*}
\rho_0(x)=\left\{
\begin{array}{ll}
g_L(x;0.0,0.1),& \qquad 0\leq x<x_J^{(L)}, \,(x_J^{(L)}\approx 0.09)\\
g_L(x;0.5,A),& \qquad  x_J^{(L)}\leq x\leq 0.5, \,(A\approx 0.703)\\
g_R(x;1.0,0.95),& \qquad  0.5 <  x \leq 1.
\end{array}
\right.
\end{equation*}

Solutions from Examples 3 and 4 are depicted in Fig.~\ref{fig:nohomog-slow}. These two examples illustrate two possibilities for slow-fast microtubules: when $\rho_0(x)$ is continuous from the right at $x_0$ (Example 3) and when it is continuous from the left at $x_0$ (Example 4).  

\section{Monte Carlo Simulations}\label{sec:num}

We now show that the results from Section~\ref{sec:nonhomo} involving the mean-field continuous model in the inhomogeneous case are consistent with those from Monte Carlo simulations.  To this end, we return to the discrete problem \eqref{original_system_stationary}-\eqref{bc} with $M=500$ lattice sites. We perform $R=10^4$ realizations and compare the resulting discrete densities with the continuum equation \eqref{mean_field_original} as $\ve \to 0$, computed by representation formulas from Theorem~\ref{thm:nonhomog1} and \ref{thm:nonhomog2}. 

\begin{figure}[t]
	\begin{center}
		\includegraphics[width=0.475\textwidth]{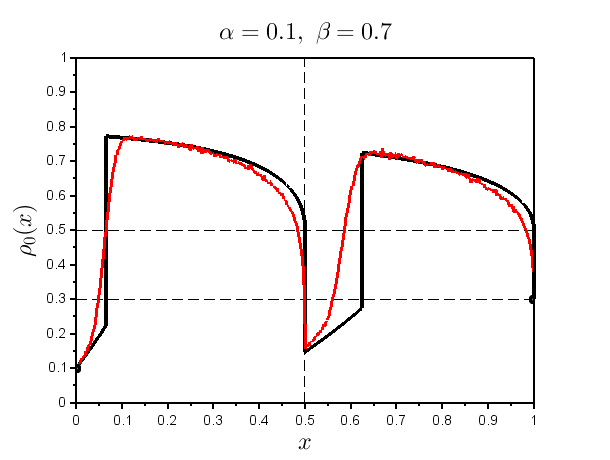}
		\includegraphics[width=0.475\textwidth]{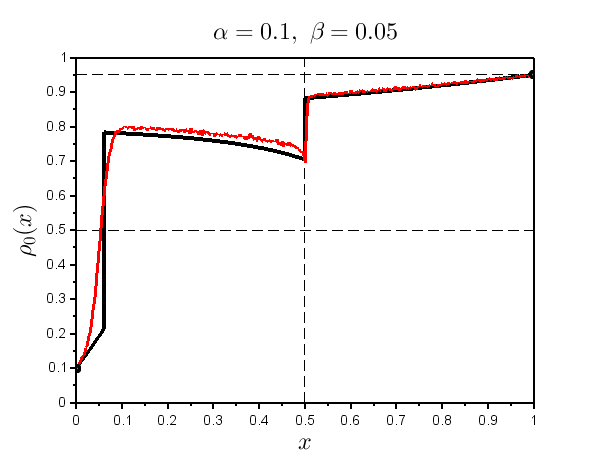}
		\caption{Examples 3 \& 4 are simulated with $R=10^4$ realizations of the lattice model, described in Section~\ref{sec:num}.  Simulations use the same parameters as Examples 3 \& 4 for inhomogeneous microtubules above (analytical solution is solid black line; see also Figure~\ref{fig:nohomog-slow}).  We see a strong quantitative agreement between the Monte Carlo (red) and analytical (black) solutions.  
		}
		\label{fig:MCnohomog-slow}
	\end{center}
\end{figure}

Specifically, we adapt a similar algorithm to the one in \cite{ParFraFre2003}. In each realization $r=1,...,R$ we consider tuple $\{\nu^{(n)}(i)\}$ where $n$ stands for the iteration step number and $\nu^{(n)}(i)=1$ if the $i$th site is occupied at the $n$th iteration step and $\nu^{(n)}(i)=0$, if otherwise. Initially, microtubule is empty, i.e.,  $\nu^{(0)}(i)=0$ for $i=1,...,M$. For each time step,  $n=1,...,N$ discrete dynamics of $\left\{\nu^{(n)}(i)\right\}_{i=1}^{M}$ is described by the following procedure:   
\begin{itemize}
	\item[{\bf 1.}] Choose randomly site $i$ (all sites are equiprobable).
	\item[{\bf 2.}] If $i=1$ and $\nu^{(n)}(1)=0$, then $\nu^{(n+1)}(1)=1$ with probability $\alpha$.
	\item[{\bf 3.}] If $i=M$ and $\nu^{(n)}(M)=1$, then $\nu^{(n+1)}(M)=0$ with probability $\beta$.
	\item[{\bf 4.}] If $1<i<M$, then 
	\begin{itemize}
		\item if $\nu^{(n)}(i)=0$, then $\nu^{(n+1)}(i)=1$ and $\nu^{(n+1)}(i-1)=0$ with probability $v_{i-\frac{1}{2}}$ provided that $\nu^{(n)}(i-1)=1$;
		\item if $\nu^{(n)}(i)=1$, then $\nu^{(n+1)}(i)=0$ and $\nu^{(n+1)}(i+1)=1$ with probability $v_{i+\frac{1}{2}}$ provided that $\nu^{(n)}(i+1)=0$.
	\end{itemize}
	\item[{\bf 5.}] If $1<i<M$, then 
		\begin{itemize}
			\item if $\nu^{(n+1)}(i)=1$ after step 4, then let $\nu^{(n+1)}(i)=0$ with probability $\omega_d$;
			\item if $\nu^{(n+1)}(i)=0$ after step 4, then let $\nu^{(n+1)}(i)=1$ with probability $\omega_a$.
		\end{itemize}
\end{itemize}
Finally, after running $N$ steps and $R$ realizations, we assign $\rho_i^{\text{MCS}}:=\langle \nu_i^{(N)}\rangle_{r}$.

Monte Carlo simulations are in very good agreement with the outer solutions derived in Section~\ref{sec:nonhomo}. Specifically, results of Monte Carlo simulations corresponding to Example 3 and 4 from Section~\ref{sec:nonhomo} are depicted in Fig.~\ref{fig:MCnohomog-slow}; one can see agreement between histograms obtained from Monte Carlo simulations (red) and analytical solutions (black).
  Observe that number of realizations $R$ and number of iteration steps $N$ needed to reach equilibrium are critical for recovering the sharp transitions observed near the interface of the inhomogeneous microtubule. For Monte Carlo histograms in Fig.~\ref{fig:MCnohomog-slow}, we were to take $N=2.5\cdot 10^{6}$ time steps in order to guarantee that a transient solution has reached equilibrium for each of $R=10^4$ realizations.
The large numbers for both realizations and iteration steps lead to the observation that the reproduction of equilibrium profiles (solutions) in Examples 1-4 by Monte Carlo simulations to be very time consuming, whereas Theorems~\ref{thm:nonhomog1} and \ref{thm:nonhomog2} give explicit formulas which require only numerical integration of at most four ODEs for the auxiliary functions $g(x;s,a)$.

\section{Discussion}
In this work, we present a mathematical model to describe dynamics of motor proteins on microtubules. Using methods from asymptotic analysis, we provide closed-form expressions for motor protein density solutions. We also provide verification of the results of mathematical analysis by Monte Carlo simulation with the discrete MT model. The mathematical model may serve as a convenient framework for studying experimental data. Even more, the modeling and analysis may assist in inferring \textit{in vivo} dynamics where biophysical imaging is limited in the crowded cellular environments. It is also important to note that the model presented herein is consistent with prior theoretical results for the homogeneous case (e.g., \cite{ParFraFre2003,ParFraFre2004,Fre11}).

The model approach developed herein provides additional advantages over the prior approach of Frey et al. \cite{Fre18} and others \cite{Klu08,Klu05} while remaining faithful in the homogeneous case.  Most notably, the model is developed to study inhomogeneous regimes where large density profiles can result in the emergence of internal boundary layers.  Beyond the obvious application to motor protein dynamics along a microtubule, this also provides insight into traffic flow problems.  The PDE governing the density of cars has a similar form to the equation governing the density of motor proteins here.  This work also provides an additional example of the power of analyzing discrete ODE model systems by passing to the limit and obtaining a mean-field PDE.

What made this work challenging is that {\it a priori} initial data cannot predict regions of low or high density.  Even within the Monte Carlo simulations we observe that they must be run for a significant length of time to capture all the feature of the solutions (e.g, interior boundary layers, sharp transitions etc.).   An additional challenge lies in experimental verification given the current state of technology.  Once imaging technology improves combined with advancements in biophysical knowledge, the theory developed in this manuscript can be rigorously tested experimentally both {\it in vitro} and {\it in vivo}.  This will be crucial in verifying model parameter regimes corresponding to biologically realistic results.

This work lays the foundation for future work in understanding inhibited transport along microtubules.  The model we present may be augmented to account for more biological realism in describing motor protein dynamics and intracellular transport. Realistically there are several ``lanes" on these MTs which motor proteins move laterally and they may switch lanes. Similarly, motor proteins may also change directions when encountering patches \cite{Ros08}. Furthermore, transport takes place on highly complex 3-dimensional networks of many MTs and AFs. Hence modeling the intersections between MTs would be of interest as well as analyzing the composite density profiles using the analysis presented in this work. In addition, the cargo transported to and from the cell nucleus and cell wall is carried by motor proteins \cite{Lakadamyali2014,Gro07} and the given model may be augmented to account for this cargo. We also note that motor proteins transfer from MT to MT within the cell, and the model as well as analysis developed here may serve as a foundation for this study.  

Overall, the model for an inhomogeneous microtubule presented here can inform motor protein dynamics in {\it rough} regimes where transport properties are not consistent along given trajectories.  This will ultimately lay the groundwork for fundamental understanding of the onset of neurodegenerative diseases.  The inhomogeneous microtubule model may be used to investigate how one can control transport properties of motor proteins in high density regimes along microtubules.  Given the structure of a microtubule, can one devise conditions so that the equilibrium solution contains no high density regimes (jams) by understanding or imposing defects along its surface?  Also, given a distribution of inhomogeneous regions ($N > 2$) on a microtubule can we predict the equilibrium solution?  The answers to these questions may be the source of further investigation in a future work.

%
%
%

\vspace{0.2 in}

\begin{acknowledgements}
The work of SR was supported by the Cleveland State University Office of Research through a Faculty Research Development Grant.
\end{acknowledgements}

\bibliographystyle{spbasic}

\appendix

\section{Homogeneous microtubules: Proof of Theorem \ref{thm:homog} and Corollary \ref{thm2:homog}}
\label{appendix:homog}

Equation \eqref{homo_mean_field} may be rewritten in the form of a system of two first order ODEs for density $\rho_\eps$ and flux $J_\eps$ (see also \eqref{key_system}): 
\begin{equation}\label{system_ode}
\left\{\begin{array}{rl}
\dfrac{\eps}{2}\rho_\eps' &= -v_0^{-1}J_\eps +\rho_\eps(1-\rho_\eps), \\
J_\eps'&=\Omega_A-(\Omega_A+\Omega_D) \rho_\eps.
\end{array}\right.
\end{equation}

\begin{figure}[]
	\begin{center}
		\includegraphics[width=0.45\textwidth]{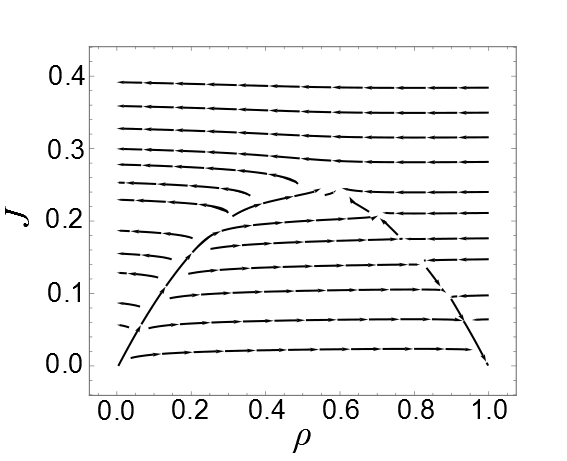}
		\includegraphics[width=0.5\textwidth]{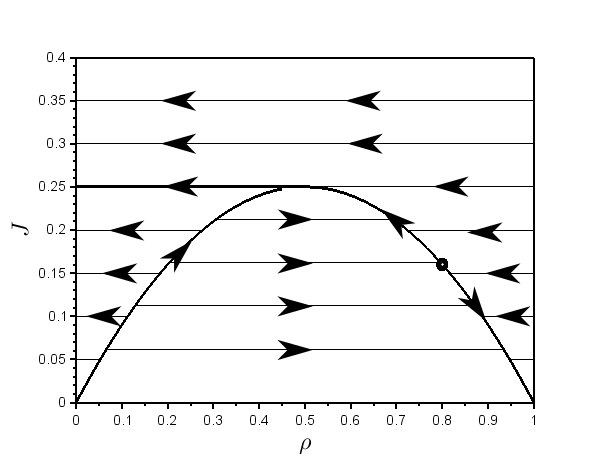}
		\caption{Left: Phase portrait for \eqref{system_ode} with $\eps=0.01$, $\Omega_A=0.7$ and $\Omega_D=0.3$; Right: Sketch of the phase portrait for \eqref{system_ode} with $\eps \ll 1$, the black circle represents the stationary point.}
		\label{fig:sketch}
	\end{center}
\end{figure}

Next, we discuss the phase portrait for this system with $\ve\ll 1$, depicted in Fig.~\ref{fig:sketch}. Away from curve $\gamma$ defined by 
\begin{equation}
\gamma:=\left\{(\rho,J)\left|J=v_0\rho(1-\rho) \text{ and }\begin{array}{l}0\leq \rho\leq 1,\\ 0\leq J \leq v_0/4\end{array}\right.\right\},
\end{equation} 
the trajectories of \eqref{system_ode}, parametrized by $0\leq x\leq \ell$, have almost horizontal slope in $(\rho,J)$ plane.  This is because the slope of $\rho_\ve$ is of the order $\ve^{-1}$, that is $\rho_\ve'(x)\sim \ve^{-1}$, whenever the point $(\rho_\ve(x),J_\ve(x))$ is away from $\gamma$ (it follows from the first equation in \eqref{system_ode}). It would be natural to expect that as $\ve$ vanishes, trajectory $\left\{(\rho_\ve(x),J_\ve(x)),0\leq x \leq \ell \right\}$  approaches the arch $\gamma$ and this trajectory is contained in a given thin neighborhood of $\gamma$ for sufficiently small $\ve$. In this subsection, it will be shown that the behavior of the solution is more complicated than simply evolving near $\gamma$.    

To describe how the solution $\rho_\ve(x)$ behaves for $\ve\ll 1$, we introduce the following notation for parts of curve $\gamma$. Namely,
\begin{eqnarray*}
	\gamma_{l}&:=&\gamma \cap \left\{0\leq \rho< 0.5 \right\},\\
	\gamma_{r,+}&:=&\gamma \cap \left\{ 0.5\leq \rho\leq \rho_\text{eq} \right\},\\
	\gamma_{r,-}&:=&\gamma \cap \left\{ \rho_\text{eq}\leq \rho\leq 1 \right\}.
\end{eqnarray*} 
Here $\rho_\text{eq}:=\Omega_A/(\Omega_A+\Omega_D)$. 
Let us also introduce the following horizontal segment 
\begin{equation*}
\Gamma :=\left\{(\rho,J):J=v_0/4,~ 0\leq \rho\leq 0.5 \right\},
\end{equation*}
and the solution $g(x;s,a)$ to \eqref{def_of_gg}, i.e., the initial value problem of the first order obtained by the formal limit as $\ve \to 0$ in \eqref{homo_mean_field}: 
\begin{equation}
\label{def_of_g}
v_0(1-2g)\partial_x g=\Omega_A-(\Omega_A+\Omega_D)g,~~~g(s;s,a)=a.
\end{equation} 

First, note that $\gamma_l$, which is the left part of the curve $\gamma$, is unstable, that is all trajectories, excluding $\gamma_l$, are directed away from $\gamma_l$ in the vicinity of $\gamma_l$. The right part of the curve $\gamma$, consisting of curve segments $\gamma_{r,+}$ and $\gamma_{r,-}$, is stable, attracting all trajectories in its vicinity, except those that follow $\Gamma$. We note that this exception, when $\gamma_{r,+}$ loses its stability, occurs at the interface point $(\rho=1/2,J=v_0/4)$ where $\gamma_{r,+}$ meets $\gamma_l$.  All trajectories  reaching this point near (not necessarily intersecting)  the curves $\gamma_{r,+}$ and $\gamma_l$ continue along $\Gamma$. 

Given specific values of $\alpha,\beta \in (0,1)$ in boundary conditions \eqref{homo_bc}, the statement of Theorem~\ref{thm:homog} as well as representation formula \eqref{construction_of_outer_solution} can be simply verified by careful inspection of the phase portrait depicted in Fig.~\ref{fig:sketch}. Specifically, for all $0<\alpha,\beta<1$, one can draw a path $\left\{(\rho(x),J(x)):0\leq x\leq \ell\right\}$ along arrows in Fig.~\ref{fig:sketch} (right), which starts at vertical line $\rho=\alpha$ and ends at vertical line $\rho=1-\beta$, and such a path will be unique for given $\alpha$ and $\beta$ (see also left column of Fig. \ref{fig:examples} for specific examples). Instead of checking each couple $(\alpha,\beta)$, one  would split ranges of $(\alpha,\beta)$ into sub-domains within which the outer solution has constant or smoothly varying shape, as it is done in proof below. 

\medskip 

%



\noindent {\it Proof of Theorem~\ref{thm:homog}.} Consider the following functions: 
\begin{equation*}
\rho_\alpha(x)=g(x;0,\alpha) \text{ and }
\rho_\beta(x)=g(x;\ell,\max\left\{0.5,1-\beta \right\}).
\end{equation*}
These functions can be thought of as one-sided solutions (i.e., satisfying one of the boundary conditions, either $\rho(0)=\alpha$ or $\rho(\ell)=\max\left\{0.5,1-\beta\right\}$) of Equation \eqref{homo_mean_field} for $\varepsilon=0$. The reason we choose $\rho(\ell)=\max\left\{0.5,1-\beta\right\}$ instead of $\rho_\beta(\ell)=1-\beta$ is because there is no solution continuous at $x=\ell$ with $\rho(\ell)<0.5$ as visible in Fig.~\ref{fig:sketch} (curve $\gamma$ is unstable in region $\left\{0\leq \rho< 0.5 \right\}$).  

Introduce also the corresponding fluxes: 
\begin{equation*}
J_\alpha(x)=v_0\rho_\alpha(x) (1-\rho_{\alpha}(x))\text{ and }J_{\beta}(x)=v_0\rho_\beta(x) (1-\rho_{\beta}(x)).
\end{equation*}
From the definition of function $g$ it follows that $J_\alpha(x)$ and $J_\beta(x)$ are both monotonic functions, and function $J_\beta(x)$ is defined for all $0\leq x < \ell$. Moreover, $J_\beta(x)$ can be extended onto $(-\infty,\ell]$ and 
\begin{equation*}
\lim\limits_{x\to-\infty}J_{\beta}(x)=J_\text{eq}, 
\text{ where }J_\text{eq}:=v_0\dfrac{\Omega_A\Omega_D}{(\Omega_A+\Omega_D)^2}.
\end{equation*}  

Consider case $\alpha\geq 0.5$. From Fig.~\ref{fig:sketch}, it follows that a trajectory emanating for initial point $(\alpha,J)$ for any $0<J<v_0/4$ 
immediately reaches $\gamma_r$ and stays on $\gamma_r\cup \Gamma$ for $0<x\leq \ell$. Thus, at $x=0$ trajectory $\left\{(\rho_0(x),J_0(x)): 0\leq x \leq \ell\right\}$, describing the outer solution, jumps from $(\alpha,J_0(0))$ at $t=0$  to $\gamma_r$:
\begin{equation}
\rho_0(x)=
\left\{
\begin{array}{ll} 
\alpha, & x=0,\\
\rho_\beta(x),& 0< x \leq \ell.
 \end{array}
 \right.
\label{alpha_more_12}
\end{equation}

In the case where $\alpha< 0.5$, denote by $0\leq x_J\leq \ell$ location at which fluxes $J_\alpha(x)$ and $J_\beta(x)$ intersect, that is, 
\begin{equation}\label{equality_of_fluxes_alpha_and_bet}
J_\alpha(x_J)=J_\beta(x_J).
\end{equation}
Equality \eqref{equality_of_fluxes_alpha_and_bet} implies that either $\rho_\alpha(x_J)=1-\rho_\beta(x_J)$ or $\rho_\alpha(x_J)=\rho_\beta(x_J)$. If $\rho_\alpha(x_J)=\rho_\beta(x_J)$, then since $\rho_\alpha$ and $\rho_\beta$ are solutions of the same first order ordinary differential equation, these two functions coincide $\rho_\alpha(x)\equiv\rho_\beta(x)$.  

We show now that either
\begin{equation}\label{such_xj_at_most_one}
\text{there exists at most one $x_J \leq 1$ or $\rho_\alpha(x)\equiv \rho_{\beta}(x)$.}
\end{equation}
Indeed, since $\alpha<0.5$, trajectory $(\rho_\alpha(x),J_\alpha(x))$ evolves on $\gamma_{l}$ for all $0\leq x \leq \ell$ where solution $\rho_\alpha(x)$ exists, and $J_{\alpha}(x)$ monotonically increases. Trajectory $(\rho_\beta(x),J_\beta(x))$ evolves also for all $0\leq x \leq \ell$ within either $\gamma_{r,+}$ or $\gamma_{r,-}$. If $(\rho_\beta(x),J_\beta(x))$ evolves within $\gamma_{r,-}$, then $J_\beta(x)$ is monotonically decreasing in $x$ whereas $J_\alpha(x)$ is monotonically increasing $x$, and thus equation $J_\alpha(x)=J_\beta(x)$ can have at most one root in this case. If $(\rho_\beta(x),J_\beta(x))$ evolves within $\gamma_{r,+}$, then both $J_\alpha (x)$ and $J_\beta (x)$ increase with $x$. Assume that there are at least two distinct numbers $x_J^{(1)}$, $x_J^{(2)}$ such that $x_J^{(1)}<x_J^{(2)}$ and $J_\alpha(x_J^{(i)})=J_{\beta}(x_J^{(i)})$, $i=1,2$. Assume also that $x_J^{(1)}$ and $x_J^{(2)}$ are neighbor roots of equation $J_\alpha(x)= J_{\beta}(x)$, i.e., for all $x\in (x_J^{(1)},x_J^{(2)})$ we have $J_\alpha(x)\neq J_{\beta}(x)$. Then due to  
\begin{equation*}
\partial_x J=\Omega_A-(\Omega_A+\Omega_D)g, \text{ where }J(x)=v_0 g(x)(1-g(x)) 
\end{equation*}
and $\rho_\alpha(x_J^{(i)})<0.5$ $\rho_\alpha(x_J^{(i)})>0.5$, $i=1,2$, we have that $\partial_x J_\alpha(x_J^{(i)})>\partial_x J_\beta(x_J^{(i)})$, $i=1,2$. Noting that a smooth function can't have the same sign of its derivative at two successive roots we arrive to contradiction. Therefore, such $x_J$ is at most one and \eqref{such_xj_at_most_one} is shown. 

If $J_\alpha(x)\neq J_\beta(x)$  for all $0\leq x \leq 1$, then define $x_J$ as follows: 
\begin{equation*}
x_J=\left\{
\begin{array}{ll} 
0, & J_\beta(x)<J_\alpha(x)\text{ for all }0<x<\ell, \\ 
1, & J_\alpha(x)<J_\beta(x)\text{ for all }0<x<\ell.
\end{array}
\right.
\end{equation*}
We note that point $x=x_J$ is where the outer solution jumps from $\rho_\alpha(x)$ to $\rho_\beta(x)$, thus 
\begin{equation}
\rho_0(x)=\left\{
\begin{array}{ll}
\rho_\alpha(x),&0\leq x<x_J,\\
\rho_\beta(x),&x_J<x<\ell.
\end{array}
\right.
\label{alpha_less_12}
\end{equation}
and $\rho_0(\ell)=1-\beta$.

Formulas \eqref{alpha_more_12}, \eqref{alpha_less_12}, and \eqref{homo_bc} complete the proof of Theorem~\ref{thm:homog}. 

\rightline{$\square$}

%


\section{Examples of solutions given by \eqref{construction_of_outer_solution}} 

To illustrate the result of Theorem \ref{thm:homog} we continue with the following examples. We take $v_0=1$, $\ell=1$, $\Omega_A=0.8$ and $\Omega_D=0.2$, and we vary the boundary rates $\alpha$ and $\beta$. 
 The outer solution for each example, as both a trajectory in $(\rho, J)$ plane  and the plot of $\rho_0(x)$, is depicted in Fig.~\ref{fig:examples}.

\begin{figure}
	\begin{center}
		\includegraphics[width=0.45\textwidth]{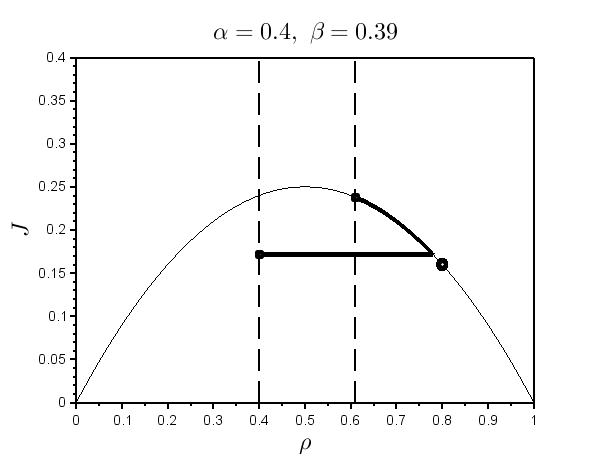}
		\includegraphics[width=0.45\textwidth]{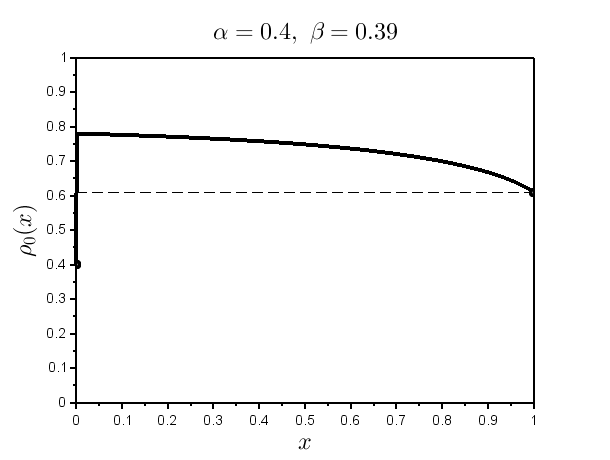}
		\includegraphics[width=0.45\textwidth]{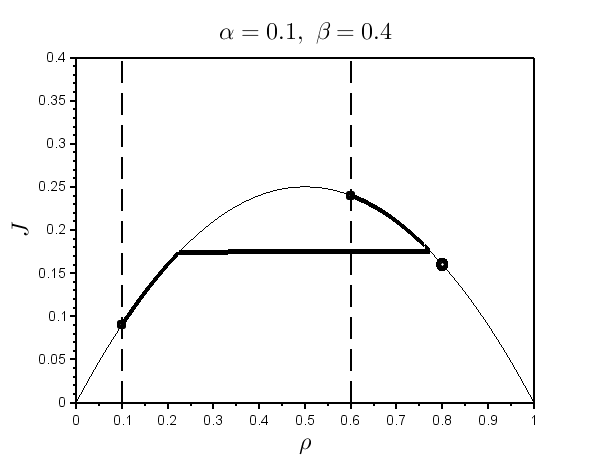}
		\includegraphics[width=0.45\textwidth]{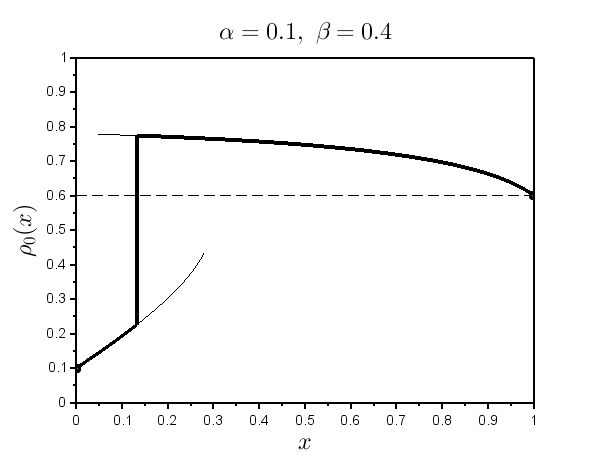}
		\includegraphics[width=0.45\textwidth]{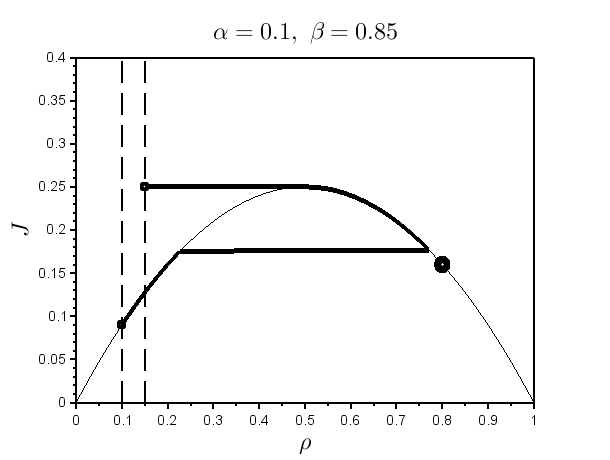}
		\includegraphics[width=0.45\textwidth]{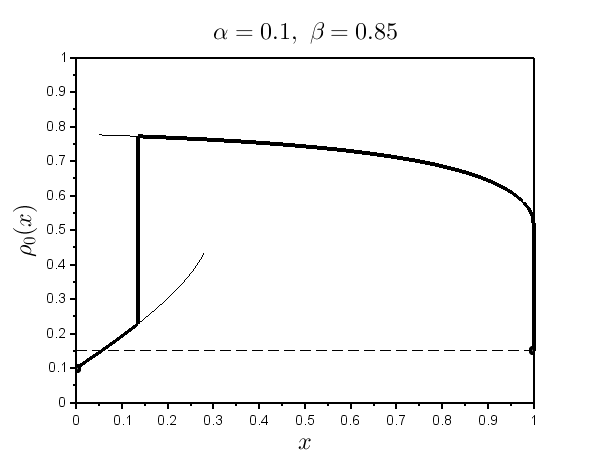}
		\includegraphics[width=0.45\textwidth]{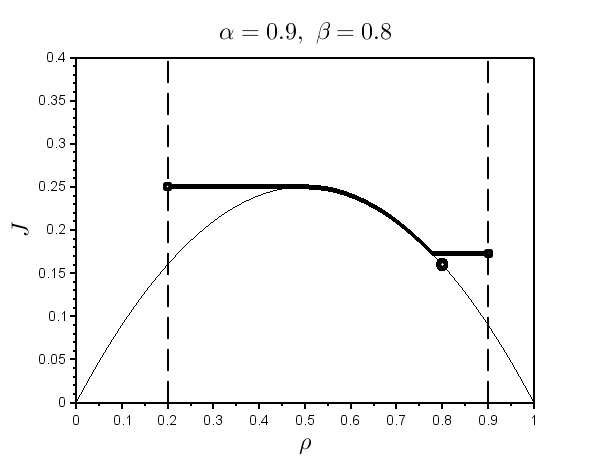}
		\includegraphics[width=0.45\textwidth]{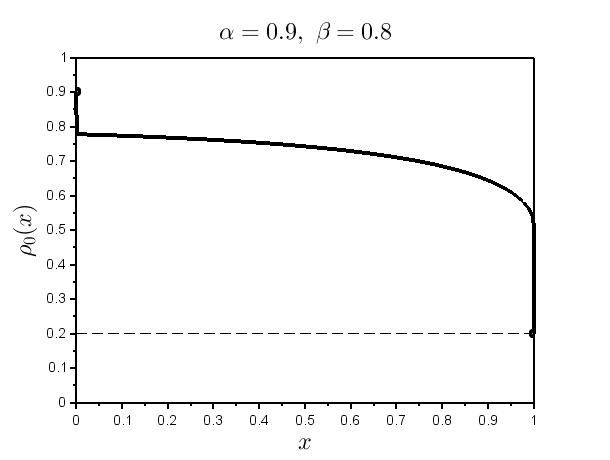}
		\caption{Left: The thick line represents the trajectories from Examples 1-4; it starts at $\rho=\alpha$ and ends at $\rho=1-\beta$, the black circle at (0.8,0.16) represents the stationary solution. Right: The thick line represents the outer solution $\rho_0(x)$ for Examples 1-4. In Examples 2 and 3, branches $g(x;0,\alpha)$ and $g(x;1,\max\{0.5,1-\beta\})$ extend slightly  beyond the intervals where they are a part of the outer solution $\rho_0(x)$ (thin curves).}
		\label{fig:examples}
	\end{center}
\end{figure}

\smallskip 

\noindent{\bf Example 1.} $\alpha=0.4$ and $\beta=0.39$.
\begin{equation*}
\rho_0(x)=\left\{\begin{array}{ll}0.4,&x=0\\ g(x;1,0.61),&0<x\leq 1.\end{array}\right.
\end{equation*}

\smallskip 

\noindent{\bf Example 2.} $\alpha=0.1$ and $\beta=0.4$.
\begin{equation*}
\rho_0(x)=\left\{\begin{array}{ll}
g(x;0,0.1),& 0\leq x \leq x_J,\,x_J\approx 0.133\\
g(x;1,0.6),& x_J< x \leq 1.
\end{array}\right.
\end{equation*}  

\smallskip

\noindent{\bf Example 3.} $\alpha=0.1$ and $\beta=0.85$.

\begin{equation*}
\rho_{0}(x)=\left\{
\begin{array}{ll} g(x;0,0.1),& 0\leq x \leq x_J,\, x_J\approx 0.135,\\ g(x;1,1/2),& x_J<x<1,\\0.15,& x=1.\end{array}\right.
\end{equation*}

\smallskip 

\noindent{\bf Example 4.} $\alpha=0.9$ and $\beta =0.8$.
\begin{equation*}
\rho_{0}(x)=\left\{
\begin{array}{ll} 
0.9,& x=0,\\ 
g(x;1,1/2),& 0<x<1,\\
0.2,& x=1.
\end{array}
\right.
\end{equation*}


The case $x_J>1$ corresponds to the case of fast motor proteins or, more precisely, unidirectional motion dominates attachment/detachment, and thus resulting density is low in MT, $\rho_0(x)<0.5$ for $x\in (0,1)$. Consider the following example: 

\smallskip 

\noindent {\bf Example 5.} $\alpha=0.05$, $\beta=0.85$, $\Omega_A=0.16$ and $\Omega_D=0.04$.
\begin{equation*}
\rho_0(x)=\left\{
\begin{array}{ll} 
g(x,0,\alpha), & 0\leq x<1,\\
1-\beta, & x=1. 
\end{array}
\right.
\end{equation*}
The solution is depicted in Fig.~\ref{fig:fast}.

\begin{figure}
	\includegraphics[width=0.45\textwidth]{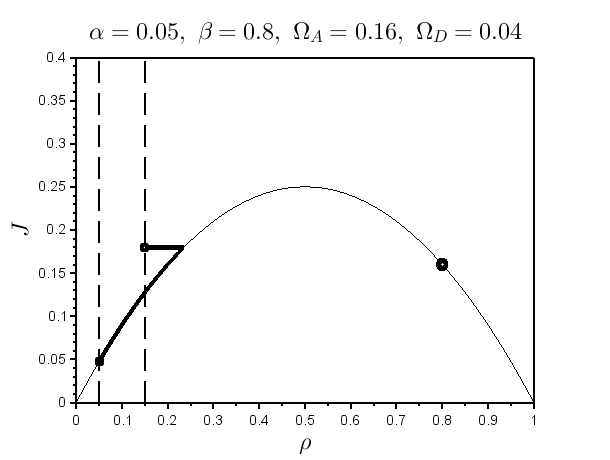}
	\includegraphics[width=0.45\textwidth]{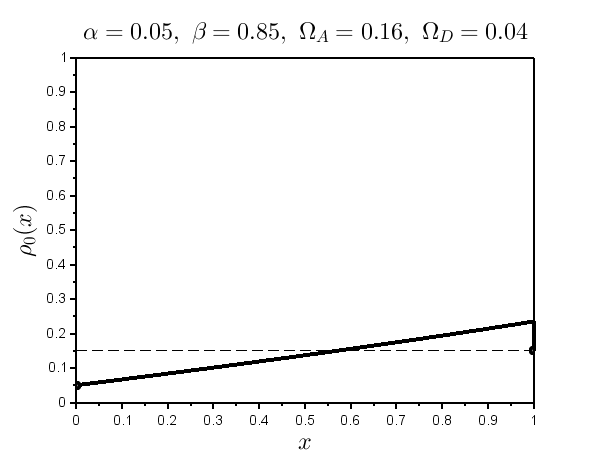}
	\caption{The thick line represents the trajectories from Example 5; it starts at $\rho=\alpha$ and ends at $\rho=1-\beta$, the black circle at (0.8,0.16) represents the stationary solution. Right: The thick line represents the outer solution $\rho_0(x)$ for Example 5.}
	\label{fig:fast}
\end{figure}

 \end{document}